%% file: main.tex
\documentclass[11pt]{article}

\usepackage{indentfirst}
\usepackage{tikz, comment}
\usepackage{algpseudocode, algorithm}
\usepackage{tcolorbox, qtree}
\usepackage{fancyvrb}

\usepackage{enumerate}
\usepackage[margin=1in]{geometry}
\usepackage{amsthm,amsmath}
\usepackage{amssymb}
\usepackage{bbm}
\usepackage{graphicx}
\usepackage{subfigure}

\newtheorem{lemma}{Lemma}
\newtheorem{definition}{Definition}
\newtheorem{theorem}{Theorem}
\newtheorem{corollary}{Corollary}

\newcommand{\eat}[1]{}

\renewcommand {\S}{\mathcal S}

\renewcommand{\L}{\mathbb{L}}

\begin{document}


\title{Sequential Deliberation for Social Choice} 
\author{Brandon Fain\thanks{Department of Computer Science, Duke University, 308 Research Drive, Durham, NC 27708. {\tt btfain@cs.duke.edu}. Supported by NSF grants CCF-1408784, CCF-1637397,  and IIS-1447554.}  
\and Ashish Goel\thanks{Supported by the Army Research Office Grant No. 116388, the Office of Naval Research Grant No. 11904718, and by the Stanford Cyber Initiative. Author's Address: Management Science and Engineering Department, Stanford University, Stanford CA 94305. Email: {\tt ashishg@stanford.edu}} \and
Kamesh Munagala\thanks{Department of Computer Science, Duke University, Durham NC 27708-0129. {\tt kamesh@cs.duke.edu}. Supported by NSF  grants CCF-1408784, CCF-1637397, and IIS-1447554.} \and
Sukolsak Sakshuwong\thanks{Management Science and Engineering Department, Stanford University, 353 Serra Mall, Stanford, CA 94305. {\tt sukolsak@stanford.edu}}}
\date{}



\maketitle 

\begin{abstract}

In large scale collective decision making, social choice is a normative study of how one ought to design a protocol for reaching consensus. However, in instances where the underlying decision space is too large or complex for ordinal voting, standard voting methods of social choice may be impractical.  How then can we design a mechanism - preferably decentralized, simple, scalable, and not requiring any special knowledge of the decision space - to reach consensus?  We propose sequential deliberation as a natural solution to this problem. In this iterative method, successive pairs of agents bargain over the decision space using the previous decision as a disagreement alternative. We describe the general method and analyze the quality of its outcome when the space of preferences define a median graph. We show that sequential deliberation finds a 1.208-approximation to the optimal social cost on such graphs, coming very close to this value with only a small constant number of agents sampled from the population. We also show lower bounds on simpler classes of mechanisms to justify our design choices. We further show that sequential deliberation is ex-post Pareto efficient and has truthful reporting as an equilibrium of the induced extensive form game. We finally show that for general metric spaces, the second moment of of the distribution of social cost  of the outcomes produced by sequential deliberation is also bounded.

\end{abstract}


\section{Introduction}
\input{introduction.tex}

\input{results.tex}

\input{properties.tex}
\input{generalMetrics.tex}

\input{open.tex}
\newpage

\bibliographystyle{splncs03}
\bibliography{refs.bib}
\appendix
\input{appendix.tex}

\end{document}

%% file: introduction.tex
Suppose a university administrator plans to spend millions of dollars to update her campus, and she wants to elicit the input of students, staff, and faculty.  In a typical social choice setting, she could first elicit the bliss points of the students, say ``new gym,'' ``new library,'' and ``new student center.''   However, voting on these options need not find the social optimum, because it is not clear that the social optimum is even on the ballot.  In such a setting, \textit{deliberation} between individuals would find entirely new alternatives, for example ``replace gym equipment plus remodeling campus dining plus money for scholarship.'' This leads to finding a social optimum over a wider space of semi-structured outcomes that the system/mechanism designer was not originally aware of, and the participants had not initially articulated.

We therefore start with the  following premise: The mechanism designer may not be able to enumerate the outcomes in the decision space or know their structure, and this decision space may be too big for most ordinal voting schemes. (For instance, ordinal voting is difficult to implement in complex combinatorial spaces~\cite{Xia-chapter} or in continuous spaces~\cite{GargKGMM17}.) However, we assume that agents can still reason about their preferences and small groups of agents can negotiate over this space and collaboratively propose outcomes that appeal to all of them. Our goal is to design protocols based on such a primitive by which small group negotiation can lead to an aggregation of societal preferences without a need to formally articulate the entire decision space and without every agent having to report ordinal rankings over this space.

The need for small groups is motivated by a practical consideration as well as a theoretical one. On the practical side, there is no online platform, to the best of our knowledge, that has a successful history of large scale deliberation and decision making on complex issues; in fact, large online forums typically degenerate into vitriol and name calling when there is substantive disagreement among the participants. Thus, if we are to develop practical tools for decision making at scale, a sequence of small group deliberations appears to be the most plausible path. On the theoretical side, we understand the connections between sequential protocols for deliberation and axiomatic theories of bargaining for small groups, e.g. for pairs~\cite{RubinsteinBargaining,Rubinstein2}, but not for large groups, and we seek to bridge this gap.

\paragraph{Summary of Contributions.} Our main contributions in this paper are two-fold:
\begin{itemize}
\item A simple and practical sequential protocol that only requires agents to negotiate in pairs and collaboratively propose outcomes that appeal to both of them. 
\item A canonical analytic model in which we can precisely state properties of this protocol in terms of approximation of the social optimum, Pareto-efficiency, and incentive-compatibility, as well as compare it with simpler protocols.
\end{itemize}  


\subsection{Background: Bargaining Theory} 
Before proceeding further, we review bargaining, the classical framework for two-player negotiation in Economics. Two-person bargaining, as framed in~\cite{NashBargaining}, is a game wherein there is a disagreement outcome and two agents must cooperate to reach a decision; failure to cooperate results in the adoption of the disagreement outcome.  Nash postulated four axioms that the bargaining solution ought to satisfy assuming a convex space of alternatives: Pareto optimality (agents find an outcome that cannot be simultaneously improved for both of them), symmetry between agents, invariance with respect to affine transformations of utility (scalar multiplication or additive translation of any agent's utility should not change the outcome), and independence of irrelevant alternatives (informally that the presence of a feasible outcome that agents do not select does not influence their decision). Nash proved that the solution maximizing the Nash product (that we describe later) is the unique solution satisfying these axioms.  To provide some explanation of how two agents might find such a solution, ~\cite{RubinsteinBargaining} shows that Nash's solution is the subgame perfect equilibrium of a simple repeated game on the two agents, where the agents take turns making offers, and at each round, there is an exogenous probability of the process terminating with no agreement.

The two-person bargaining model is therefore clean and easy to reason about. As a consequence, it has been extensively studied. In fact, there are other models and solutions to two-person bargaining, each with a slightly different axiomatization~\cite{KalaiProportionalSolutions,KalaiSmorodinskyBargaining,MyersonComparableUtility}, as well as several experimental studies~\cite{RothBargaining,NeelinBargaining,BinmoreExperiment}. In a social choice setting, there are typically many more than two agents, each agent having their own complex preferences. Though bargaining can be generalized to $n$ agents with similar axiomatization and solution structure, such a generalization is considered impractical. This is because in reality it is difficult to get a large number of individuals to negotiate coherently; complexities come with the formation of coalitions and power structures~\cite{Harsanyi,Krishna}. Any model for simultaneous bargaining, even with three players~\cite{OddManOut}, needs to take these messy aspects into account.

\subsection{A Practical Compromise: Sequential Pairwise Deliberation}
In this paper, we take a middle path, avoiding both the complexity of explicitly specifying preferences in a large decision space that any individual agent may not even fully know (fully centralized voting), and that of simultaneous $n$-person bargaining (a fully decentralized cooperative game). We term this approach {\em sequential deliberation}.  We use 2-person bargaining as a basic primitive, and view deliberation as a sequence of pairwise interactions that refine good alternatives into better ones as time goes by.  

\renewcommand{\H}{\mathcal{H}}
\newcommand{\B}{\mathcal{B}}
\renewcommand{\L}{\mathcal{L}}
\renewcommand{\S}{\mathcal{S}}
\newcommand{\N}{\mathcal{N}}

More formally, there is a decision space $\S$ of feasible alternatives (these may be projects, sets of projects, or continuous allocations) and a set $\N$ of agents.  We assume each agent has a hidden cardinal utility for each alternative.  We encapsulate deliberation as a sequential process. The framework that we analyze in the rest of the paper is captured in Figure~\ref{fig:seq-del}.

\begin{figure}[htbp]
\centerline{
\fbox{
\parbox{\linewidth}{
\begin{enumerate}
\item In each round $t = 1,2,\ldots, T$:
\begin{enumerate}
\item A pair of agents $u^t$ and $v^t$ are chosen independently and uniformly at random with replacement. 
\item These agents are presented with a disagreement alternative $a^t$, and perform bargaining, which is encoded as a function $\B(u, v, a)$ as described below.
\item Agents $u_t$ and $v_t$ are asked to output a consensus  alternative; if they fail to reach a consensus then the alternative $a^t$ is output.
\item Let $o^t$ denote the alternative that is output in round $t$. We set $a^{t+1} = o^t$, where we assume $a^1$ is the bliss point of an arbitrary agent. 
\end{enumerate}
\item The final social choice is $a^T$. Note that this is equivalent to drawing an outcome at random from the distribution generated by repeating this protocol several times.
\end{enumerate}
}}}
\caption{\label{fig:seq-del} A framework for sequential pairwise deliberation.}
\end{figure}

Our framework is simple with low cognitive overhead, and is easy to implement and reason about. Though we don't analyze other variants in this paper, we note that the framework is flexible. For instance, the bargaining step can be replaced with any function $\B(u, v, a)$ that corresponds to an interaction between $u$ and $v$ using $a$ as the disagreement outcome; we assume that this function maximizes the Nash product, that is, it corresponds to the Nash bargaining solution. Similarly, the last step of social choice could be implemented by a central planner based on the distribution of outcomes produced.


\subsection{Analytical Model: Median Graphs and Sequential Nash Bargaining}
\label{sec:model1}
The framework in Figure~\ref{fig:seq-del} is well-defined and practical irrespective of an analytical model. However, we provide a simple analytical model for specifying the preferences of the agents in which we can precisely quantify the behavior of this framework as justification. 

\paragraph{Median Graphs.} We assume that the set $\S$ of alternatives are vertices of a {\em median graph}. A median graph has the property that for each triplet of vertices $u,v,w$, there is a unique point that is common to the three sets of shortest paths (since there may be multiple pairwise shortest paths), those between $u,v$, between $v,w$, and between $u,w$. This point is the unique {\em median} of $u,v,w$. We assume each agent $u$ has a bliss point $p_u \in \S$, and his disutility for an alternative $a \in \S$ is simply $d(p_u,a)$, where $d(\cdot)$ is the shortest path distance function on the median graph. (Note that this disutility can have an agent-dependent scale factor.)   Several natural graphs are median graphs, including  trees, points on the line, hypercubes, and grid graphs in arbitrary dimensions~\cite{Knuth}.  As we discuss in Section~\ref{sec:related}, because of their analytic tractability and special properties, median graphs have been extensively studied as structured models for spatial preferences in voting theory. Some of our results generalize to metric spaces beyond median graphs; see Section~\ref{sec:general} and Appendix~\ref{sec:budget}. 

\paragraph{Nash Bargaining.} The model for two-person bargaining is simply the classical {\em Nash bargaining} solution described before. Given a disagreement alternative $a$, agents $u$ and $v$ choose that alternative $o\in \S$ that maximizes:
$$ \mbox{Nash product } = \left( d(p_u,a) - d(p_u,o) \right) \times   \left( d(p_v,a) - d(p_v,o) \right) $$
subject to individual rationality, that is, $d(p_v,o) \leq d(p_v,a)$ and $d(p_u,o) \leq d(p_u,a)$. The Nash product maximizer need not be unique; in the case of ties we postulate that agents select the outcome that is closest to the disagreement outcome.  As mentioned before, the Nash product is a widely studied axiomatic notion of pairwise interactions, and is therefore a natural solution concept in our framework. 

\paragraph{Social Cost and Distortion.} The {\em social cost} of an alternative $a \in \S$ is given by $ SC(a) = \sum_{u \in \N} d(p_u,a)$. Let $a^* \in \S$ be the minimizer of social cost, {\em i.e.}, the {\em generalized median}. We measure the Distortion of outcome $a$ as 
\begin{equation}
\label{eq:Distortion1} 
\mbox{Distortion}(a) = \frac{SC(a)}{SC(a^*)}
\end{equation}
where we use the expected social cost if $a$ is the outcome of a randomized algorithm.

\medskip
Note that our model is fairly general. First, the bliss points of the agents in $\N$ form an arbitrary subset of $\S$. Further, the alternative chosen by bargaining need not correspond to any bliss point, so that pairs of agents are exploring the entire space of alternatives when they bargain, instead of just bliss points.  Assuming that disutility is some metric over the space follows recent literature~\cite{anshelevich2016randomized,anshelevich2015approximating,boutilier2015optimal,YuCheng,Anilesh}, and our tightest results are for median graphs specifically. 

\subsection{Our Results}
Before presenting our results, we re-emphasize that while we present analytical results for sequential deliberation in specific decision spaces, the framework in Figure~\ref{fig:seq-del} is well defined regardless of the underlying decision space and the mediator's understanding of the space.  At a high level, this flexibility and generality in practice are its key advantages.

\paragraph{Bargaining and Medians.} We first show in Section~\ref{sec:model} that on a median graph, Nash bargaining between agents $u$ and $v$ using disagreement outcome $a$ outputs the median of $p_u,p_v,a$. Therefore, $\B(u,v,a) = \mbox{Median}(p_u,p_v,a)$. On a general metric space, we show in Section~\ref{sec:second} that the Nash Bargaining outcome would lie on the shortest path between the agents, and the distance from an agent is proportional to its distance to the disagreement outcome. In a sense, agents only need to explore options on the shortest path between them.

\paragraph{Bounding Distortion.} Our main result in Section~\ref{sec:main} shows that for sequential Nash bargaining on a median graph, the expected Distortion of outcome $a^T$ has an upper bound approaching $1.208$ as $T \rightarrow \infty$. Surprisingly, we show that in $T =  \log_2\frac{1}{\epsilon} + 2.575$ steps, the expected Distortion is at most $1.208 + \epsilon$, independent of the number of agents, the size of the median space, and the initial disagreement point $a^1$. For instance, the Distortion falls below $1.22$ in at most $9$ steps of deliberation, which only requires a random sample of at most $20$ agents from the population to implement.  

In Section~\ref{sec:lb}, we ask: {\em How good is our numerical bound?} We present a sequence of lower bounds for social choice mechanisms that are allowed to use increasingly richer information about the space of alternatives on the median graph. This also leads us to make qualitative statements about our deliberation scheme.

\begin{itemize}
\item We show that any social choice mechanism that is restricted to choosing the bliss point of some agent cannot have Distortion better than $2$. 
More generally, it was recently shown~\cite{GrossAX17} that even eliciting the top $k$ alternatives for each agent does not improve the bound of $2$ for median graphs unless $k = \Omega(|\S|)$.  In effect, we show that forcing the agents to reason about cardinal utilities via deliberation leads to new alternatives that are more powerful at reducing Distortion than simply eliciting and aggregating reasonably detailed ordinal rankings.
\item Next consider mechanisms that choose, for some triplet $(u,v,w)$ of agents with bliss points $p_u, p_v, p_w$, the median outcome $m_{uvw} = \B(u,v,p_w)$. We show this has Distortion at least $1.316$, which means that sequential deliberation  is superior to one-shot deliberation that outputs $o^1$ where $a^1$ is the bliss point of some agent. 
\item Finally, for every pair of agents $(u,v)$, consider the set of  alternatives on a shortest path between $p_u$ and $p_v$. This  encodes all deliberation schemes where $\B$ finds a Pareto-efficient  alternative for some $2$ agents at each step. We show that any such mechanisms has Distortion ratio at least $9/8 = 1.125$. This space of mechanisms captures sequential deliberation, and shows that sequential deliberation is close to best possible within this space.
\end{itemize}


\paragraph{Properties of Sequential Deliberation.} We next show that sequential deliberation has several natural desiderata on median graphs in Section~\ref{sec:properties}.  In particular:

\begin{itemize}
    \item Under mild assumptions, the limiting distribution over outcomes of sequential deliberation is {\em unique}. 
    \item For every $T \ge 1$, the outcome $o^T$ of sequential deliberation is {\em ex-post Pareto-efficient}, meaning that there is no other alternative that has at most that social cost for all agents and strictly better cost for one agent. This is not a priori obvious, since the outcome at any one round only uses inputs from two agents; though it is Pareto-efficient for those two agents, it could very well be sub-optimal for the other agents.
    \item Interpreted as a mechanism, truthful play is a {\em sub-game perfect Nash equilibrium} of sequential deliberation. More precisely, we consider a different view of the function $\B(u,v,a)$ that encodes Nash Bargaining. Suppose agents $u$ and $v$ report their bliss points, and the platform implements the function $\B$ that computes the median of $p_u, p_v$, and $a$. In a sequential setting, would any agent have incentive to misreport their bliss point so that they gain an advantage (in terms of expected distance to the final social choice) in the induced extensive form game? We show that the answer is negative -- on a median graph, truthfully reporting bliss points is a sub-game perfect Nash equilibrium of the induced extensive form game.
\end{itemize}

\paragraph{Beyond Median Graphs.} In Section~\ref{sec:general}, we consider general metric spaces. We show that the Distortion of sequential deliberation is always at most a factor of $3$. 
More surprisingly, we show that sequential deliberation has constant distortion even for the second moment of the distribution of social cost of the outcomes, {\em i.e.}, the latter is at most a constant factor worse than the optimum squared social cost. This has the following practical implication: A policy designer can look at the distribution of outcomes produced by deliberation, and know that the standard deviation in social cost is comparable to its expected value, which means deliberation eliminates outlier alternatives and concentrates probability mass on more central alternatives.\footnote{See also recent work by~\cite{varianceProcaccia} that considers minimizing the variance of randomized truthful mechanisms.} We also show that such a claim cannot be made for random dictatorship, whose distortion on squared social cost is unbounded. 

\subsection{Related Work} 
\label{sec:related}
While the real world complexities of the model are beyond the analytic confines of this work, deliberation as an important component of collective decision making and democracy is studied in political science.  For examples (by no means exhaustive), see \cite{DeliberativeDemocracy,DeliberativeDemocracy2}.  There is ongoing related work on Distortion of voting for simple analytical models like points in $\mathbb{R}$ \cite{facilityLocationFeldman}, and in general metric spaces~\cite{anshelevich2016randomized,anshelevich2015approximating,boutilier2015optimal,YuCheng,Anilesh}. This work focuses on optimally aggregating ordinal preferences, say the top k preferences of a voter~\cite{GrossAX17}. In contrast, our scheme elicits alternatives as the outcome of bargaining rounds that require agents to reason about cardinal preferences.. As mentioned before,  we essentially show that for median graphs, unless $k$ is very large, such deliberation has provably lower distortion than social choice schemes that elicit purely ordinal rankings.



Median graphs and their ordinal generalization, median spaces, have been extensively studied in the context of social choice. The special cases of trees and grids have been  studied as structured models for voter preferences~\cite{SchummerV,Barbera}. For general  median spaces, it is known that the Condorcet winner (that is an alternative that pairwise beats any other alternative in terms of voter preferences) is strongly related to the generalized median~\cite{BandeltB84,SabanSM12,WendellM81} -- if the former exists, it coincides with the latter. Nehring and Puppe~\cite{NehringP07} shows that any single-peaked domain which admits a non-dictatorial and neutral strategy-proof social choice function is a median space. Clearwater {\em et al.}~\cite{Clearwater15} also showed that any set of voters and alternatives on a median graph will have a Condorcet winner. In a sense, these are the largest class of structured and spatial preferences where ordinal voting over the entire space of alternatives leads to a ``clear winner" even by pairwise comparisons. Our work stems from the assumption that this space may not be fully known to the mechanism designer or all agents.

Our paper is inspired by the {\em triadic consensus} results of Goel and Lee~\cite{GoelLee}. In that work, the authors focus on small group interactions with the goal of reaching consensus. In their model, three people deliberate at the same time, and they choose a fourth individual to whom they grant their votes. This individual takes these votes and participates in future rounds, until all votes accumulate with one individual, who is the consensus outcome. The analysis proceeds through a median graph, on which the authors show that the Distortion of the consensus approaches $1$. However, the protocol crucially assumes individuals know the positions of other individuals, and requires the space of alternatives to coincide with the space of individuals. We make neither of these assumptions -- in our case, the space of alternatives can be much larger than the number of agents, and further, individuals interact with others only via bargaining. This makes our protocol more practical, but at the same time, restricts our Distortion to be bounded away from $1$.  

The notion of {\em democratic equilibrium}~\cite{Hylland,GargKGMM17} considers social choice mechanisms in continuous spaces where individual agents with complex utility functions perform update steps inspired by gradient descent, instead of ordinal voting on the entire space. However, these schemes do not involve deliberation between agents and have little formal analysis of convergence.  Several works have considered {\em iterative voting} where the current alternative is put to vote against one proposed by different random agent chosen each step~\cite{IteratedMajority,convergenceIterative,equilibriumIterative}, or other related schemes~\cite{convergencePlurality}. In contrast with our work, these protocols are not deliberative and require voting among several agents each step; furthermore, the analysis focuses on convergence to an equilibrium instead of welfare or efficiency guarantees.

%% file: results.tex
\section{Median Graphs and Nash Bargaining}
\label{sec:model}
In this section we will use the notation $\mathcal{N}$ for a set of agents, $\mathcal{S}$ for the space of feasible alternatives, and $\mathcal{H}$ for a distribution over $\mathcal{S}$.  Most of our results are for the analytic model given earlier wherein the set $\S$ of alternatives are vertices of a {\em median graph}; see Figure~\ref{figure:medianGraph} for some examples.
\begin{definition}
A median graph $G(\S,E)$ is an unweighted and undirected graph  with the following property:  For each triplet of vertices $u,v,w \in \S \times \S \times \S$, there is a unique point that is common to the shortest paths (which need not be unique between a given pair) between $u,v$, between $v,w$, and between $u,w$. This point is the unique {\em median} of $u,v,w$. 
\end{definition}

\begin{figure}[!h]
\centering
\begin{tikzpicture} 

\fill[black] (-4,0) circle (0.1);
\fill[black] (-3,0) circle (0.1);
\fill[black] (-2,0) circle (0.1);
\fill[black] (-1,0) circle (0.1);
\draw[black, thick] (-4,0) -- (-3,0) -- (-2,0) -- (-1,0);

\fill[black] (1,1) circle (0.1);
\fill[black] (0,0) circle (0.1);
\fill[black] (1,0) circle (0.1);
\fill[black] (2,0) circle (0.1);
\draw[black, thick] (1,1) -- (0,0);
\draw[black, thick] (1,1) -- (1,0);
\draw[black, thick] (1,1) -- (2,0);

\fill[black] (3,0) circle (0.1);
\fill[black] (3,1) circle (0.1);
\fill[black] (4,0) circle (0.1);
\fill[black] (4,1) circle (0.1);
\draw[black, thick] (3,0) -- (3,1) -- (4,1) -- (4,0) -- (3,0);
\fill[black] (5,0) circle (0.1);
\fill[black] (5,1) circle (0.1);
\draw[black, thick] (4,0) -- (5,0) -- (5,1) -- (4,1);

\fill[gray] (2,1.75) circle (0.0);

\end{tikzpicture}
\caption{Examples of Median Graphs}
\label{figure:medianGraph}
\end{figure}
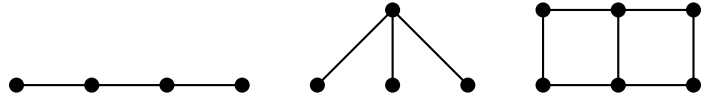

In the framework of Figure~\ref{fig:seq-del}, we assume that at every step, two agents perform Nash bargaining with a disagreement alternative.  The first results characterize Nash bargaining on a median graph.  In particular, we show that Nash bargaining at each step will select the median of bliss points of the two agents and the disagreement alternative.  After that, we show that we can analyze the Distortion of sequential deliberation on a median graph by looking at the embedding of that graph onto the hypercube.

\begin{lemma}
\label{lemma:median}
For any median graph $G=(\S,E)$, any two agents $u,v$ with bliss points $p_u, p_v \in \S$, and any disagreement outcome $a \in \S$, let $M$ be the median.  Then $M$ maximizes the Nash product of $u$ and $v$ given $a$, and is the maximizer closest to $a$.
\end{lemma}
\begin{proof}
Since $G$ is a median graph, $M$ exists and is unique; it must by definition be the intersection of the three shortest paths between $(p_u, p_v), (p_u, a), (p_v, a)$.  Note that we can therefore write $d(p_u, a) = d(p_u, M) + d(M, a)$ and similarly for $d(p_v, a)$.  
Let $\alpha = d(p_u,M)$; $\beta =  d(p_v,M)$; and $\gamma = d(a,M)$. Suppose Nash bargaining finds an outcome $o^* \in \S$. Let $x = d(o^*,p_u)$ and $y = d(o^*,p_v)$. Observing that $M$ lies on the shortest path between $p_u$ and $p_v$, and using the triangle inequality, we obtain that $ x+y \ge \alpha+\beta \implies \beta \leq x + y - \alpha$.

Noting that $d(p_u,a) = \alpha+\gamma$ and $d(p_v,a) = \beta + \gamma$, the Nash product of the point $o^*$ is:
$$ (\alpha + \gamma - x) \times ( \beta + \gamma - y) \le  (\alpha + \gamma - x) (\gamma -(\alpha - x)) = \gamma^2 - (\alpha-x)^2$$
This is maximized when $x = \alpha$ and $y = \beta$. One possible maximizer is therefore the point $o^* = M$. Suppose $d(o^*,a) < \gamma$, then by the triangle inequality, $d(p_u,o^*) > \alpha$, and similarly $d(p_v,o^*) > \beta$. Therefore, there cannot be a closer maximizer of the Nash product to $a$ than the point $M$.

 \end{proof}

\paragraph{Hypercube Embeddings.} For any median graph $G = (\S,E)$, there is an isometric embedding $\phi:G \rightarrow Q$ of $G$ into a hypercube $Q$~\cite{Knuth}. This embedding maps vertices $\S$ into a subset of vertices of $Q$ so that all pairwise distances between vertices in $\S$ are preserved by the embedding.  A simple example of this embedding for a tree is shown in Figure~\ref{figure:hypercubeEmbedding}. We use this embedding to show the following result, in order to simplify subsequent analysis.

\begin{figure}[!h]
\centering
\begin{tikzpicture} 
\draw[black, thick] (1,1) -- (0,0);
\draw[black, thick] (1,1) -- (1,0);
\draw[black, thick] (1,1) -- (2,0);
\fill[red] (1,1) circle (0.1);
\fill[green] (0,0) circle (0.1);
\fill[blue] (1,0) circle (0.1);
\fill[brown] (2,0) circle (0.1);

\draw[black, very thick] (2.25, 0.5) -- (2.75, 0.5) -- (2.65, 0.6) -- (2.75, 0.5) -- (2.65, 0.4);

\draw[black, thick] (3,0) -- (3,1) -- (4,1) -- (4,0) -- (3,0);
\draw[black, thick] (3,1) -- (3.25,1.25) -- (4.25,1.25) -- (4.25,0.25);
\draw[black, thick] (4,1) -- (4.25,1.25);
\draw[black, thick] (4,0) -- (4.25, 0.25);
\fill[green] (3,0) circle (0.1);
\fill[red] (3,1) circle (0.1);
\fill[black] (4,0) circle (0.1);
\fill[blue] (4,1) circle (0.1);
\fill[brown] (3.25,1.25) circle (0.1);
\fill[black] (4.25,0.25) circle (0.1);
\fill[black] (4.25,1.25) circle (0.1);

\fill[gray] (2,1.75) circle (0.0);

\end{tikzpicture}
\caption{The hypercube embedding of a 4-vertex star graph}
\label{figure:hypercubeEmbedding}
\end{figure}
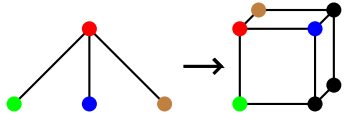

\begin{lemma}
\label{lemma:embed}
Let $G(\S,E)$ be a median graph, and let $\phi$ be its isometric embedding into hypercube $Q(V,E')$. For any three points $t,u,v \in \S$, let $M_G$ be the median of vertices $t,u,v$ and let $M_Q$ be the median of vertices $\phi(t),\phi(u),\phi(v) \in V$.  Then $\phi(M_G) = M_Q$.
\end{lemma}

\begin{proof}
By definition, since $\phi$ is an isometric embedding \cite{Knuth}, 
\begin{equation} \label{eq:isometric} d(x,y) = d(\phi(x), \phi(y)) \mbox{ for all } x,y \in \S \end{equation}
Since $G$ is a median graph, $M_G$ is the unique median of $t,u,v \in \S$, which by definition satisfies the equalities:
\begin{equation*}
    \begin{split}
        & d(t,u) = d(t,M_G) + d(M_G,u) \\
        & d(t,v) = d(t,M_G) + d(M_G,v) \\
        & d(u,v) = d(u,M_G) + d(M_G,v)
    \end{split}
\end{equation*}

$Q(V,E')$ is a hypercube, and is thus also a median graph, so $M_Q$ is the unique median of $\phi(t),\phi(u),\phi(v) \in V$, which by definition satisfies the equalities
    
    \renewcommand{\labelenumi}{\Roman{enumi}}
    \begin{enumerate}
        \item $d(\phi(t),\phi(u)) = d(\phi(t),M_Q) + d(M_Q,\phi(u))$
        \item $d(\phi(t),\phi(v)) = d(\phi(t),M_Q) + d(M_Q,\phi(v))$
        \item $d(\phi(u),\phi(v)) = d(\phi(u),M_Q) + d(M_Q,\phi(v))$
    \end{enumerate}
Applying Equation (\ref{eq:isometric}) to the first set of equalities shows that $\phi(M_G)$ satisfies equalities I,II, and III respectively.  But $\phi(M_G) \in V$ and $M_Q$ is the unique vertex in $V$ satisfying equalities I,II, and III.  Therefore, $\phi(M_G) = M_Q$.    
\end{proof}


\section{The Efficiency of Sequential Deliberation}
\label{sec:main}
In this section, we show that the Distortion of sequential deliberation is at most $1.208$. We then show that this bound is significant, meaning that mechanisms from simpler classes are necessarily constrained to have higher Distortion values. 

\subsection{Upper Bounding Distortion}
Recall the framework for sequential deliberation in Figure~\ref{fig:seq-del} and the definition of Distortion in Equation (\ref{eq:Distortion1}). We first map the problem into a problem on hypercubes using Lemma~\ref{lemma:embed}. 


\begin{corollary}
\label{lemma:hypercube}
    Let $G = (\S,E)$ be a median graph, let $\phi:G \rightarrow Q$ be an isometric embedding of $G$ onto a hypercube $Q(V,E')$, and let $\mathcal{N}$ be a set of agents such that each agent $u$ has a bliss point $p_u \in \S$.  Then the Distortion of sequential deliberation on $G$ is at most the Distortion of sequential deliberation on $\phi(G)$ where each agent's bliss point is $\phi(p_u)$.
\end{corollary}
\begin{proof}
Fix an initial disagreement outcome $a^1 \in \S$ and an arbitrary list of $T$ pairs of agents $(u^1,v^1)$, $(u^2, v^2)$,  ..., $(u^T,v^T)$.  In round 1 bargaining on $G$, Lemma~\ref{lemma:median} implies that sequential deliberation will select $o^1 = \mbox{median}(a^1,p_u^1,p_v^1)$.  Furthermore, Lemma~\ref{lemma:embed} implies that if we had considered $\phi(G)$ and bargaining on $\phi(a^1), \phi(p_u^1), \phi(p_v^1)$ instead, sequential deliberation would have selected $\phi(o^1)$.  Suppose at some round $t$ that we have a disagreement outcome $a^t$.  Then the same argument yields that if $o^t$ is the bargaining outcome on $G$, $\phi(o^t)$ would have been the bargaining outcome on $\phi(G)$.  Thus, by induction, we have that if the list of outcomes on $G$ is $o^1,...,o^T$ then the list of outcomes on $\phi(G)$ is $\phi(o^1),...,\phi(o^T)$.  But recall that $\phi(\cdot)$ is an isometric embedding and the social cost of an alternative (as defined in Section~\ref{sec:model1}) is just its sum of distances to all points in $\N$, so $a^T$ and $\phi(a^T)$ have the same social cost. 
 
Furthermore, let $a^* \in \S$ denote the generalized median of $\N$. Then, $\phi(a^*)$ has the same social cost as $a^*$. This means the median of the embedding of $\N$ into $Q$ has at most this social cost, which in turn means that the Distortion of sequential deliberation in the embedding cannot decrease.
\end{proof}


Our main result in this section  shows that as $t \rightarrow \infty$, the Distortion of sequential deliberation approaches $1.208$, with the convergence rate being exponentially fast in $t$ and independent of the number of agents $|\mathcal{N}|$,  the size of the median space $|\mathcal{S}|$, and the initial disagreement point $a^1$. In particular, the Distortion is at most $1.22$ in at most $9$ steps of deliberation, which is indeed a very small number of steps.

\begin{theorem}
\label{theorem:main}
Sequential deliberation among a set $\mathcal{N}$ of agents, where the decision space $\mathcal{S}$ is a median graph, yields $\mathbb{E}[\mbox{Distortion}(a^t)] \leq  1.208 +  \frac{6}{2^t}$.
\end{theorem}
\begin{proof}
By Corollary~\ref{lemma:hypercube}, we can assume the decision space is a $D$-dimensional hypercube $Q$ so that decision points (and thus bliss points) are vectors in $\{0,1\}^D$. 
For every dimension $k$, let $f_k$ be the fraction of agents whose bliss point has a 1 in the $k$th dimension, and let $p_{u,k}$ be the 0 or 1 bit in the $k$th dimension of the bliss point $p_u$ for agent $u$. Let $a^* \in \{0,1\}^D$ be the minimum social cost decision point, i.e., $a^* := \mbox{argmin}_{a \in Q} SC(a)$. Clearly, $a^*_k$ is 1 if $f_k > 1/2$ and 0 otherwise [assume w.l.o.g. that $a^*_k = 0$ for $f_k = 1/2$], so for every dimension $k$, the total distance to $a^*_k$, summed over $\mathcal{N}$ is:
$$\sum_{u \in \mathcal{N}} |a^*_k - p_{u,k}| = |\mathcal{N}|\mbox{min}\{f_k,\, 1-f_k\}$$
    
Now, note that sequential deliberation defines a Markov chain on $Q$.  The state in a given step is just $a^t$, and the randomness is in the random draw of the two agents.  Let $\mathcal{H}^*$ be the stationary distribution of the Markov chain.  Then we can write 
$$\lim_{t \rightarrow \infty} \mathbb{E}[\mbox{Distortion}(a^t)] = \mathbb{E}_{a \in \mathcal{H}^*}\left[\mbox{Distortion}(a)\right]$$  
    
To write down the transition probabilities, we assume this random draw is two independent uniform random draws from $\mathcal{N}$, with replacement.  We also note that Lemma~\ref{lemma:median} implies that on $Q$, sequential deliberation will pick the median in every step, \textit{i.e.}, given a disagreement outcome $a^t$ and two randomly drawn agents with bliss points $p_u, p_v$, the new decision point will be $o^t = \mbox{median}(p_u,p_v,a^t)$.  On a hypercube, the median of three points is just the dimension-wise majority.  Thus, we get a $2$-state Markov chain in each dimension $k$, with transition probabilities 
$$
        \Pr[o_k^t = 1 | a_k^t = 1] = f_k^2 + 2f_k(1-f_k) \; \mbox{and} \; \Pr[o_k^t = 1 | a_k^t = 0] = f_k^2
$$
Let $\mathcal{P}_k^*  = \lim_{t \rightarrow \infty} \Pr[o_k^t = 1]$, and let $\mathcal{H}^*_k$ denote this stationary distribution for the corresponding $2$-state Markov chain. Then, 
$$\mathcal{P}^*_k = \left(2f_k - f_k^2\right)\mathcal{P}^*_k + \left( f_k^2 \right) \left(1 - \mathcal{P}^*_k\right) \qquad \Rightarrow \qquad
        \mathcal{P}^*_k = \frac{f_k^2}{1+2f_k^2-2f_k}
$$
  By linearity of expectation, the total expected distance for every dimension $k$, summed over $u \in \mathcal{N}$ to the final outcome is given by
    \begin{equation*}
    \begin{split}
        \mathbb{E}_{a_k \in \mathcal{H}^*_k}\left[\sum_{u \in \mathcal{N}} |a_k - p_{u,k}|\right] & = |\mathcal{N}|\left( \left(\frac{f_k^2}{f_k^2 + (1-f_k)^2}\right)(1-f_k) + \left(1-\frac{f_k^2}{f_k^2 + (1-f_k)^2}\right)f_k \right) \\
        & = |\mathcal{N}|\left(\frac{f_k(1-f_k)}{f_k^2 + (1-f_k)^2} \right)
    \end{split}
    \end{equation*}
    Without loss of generality, let $f_k \in [0,1/2]$ so that for dimension $k$, the total distance to $a^*_k$ is $|\mathcal{N}|f_k$.  Then the ratio of the expected total distance to $\mathcal{H}^*_k$ to the total distance to $a^*_k$ is at most:
    \begin{equation*}
        \frac{\mathbb{E}_{a_k \in \mathcal{H}^*_k}\left[ \sum_{u \in \mathcal{N}} |p_{u,k} - a_k|\right]}{\sum_{u \in \mathcal{N}} |p_{u,k} - a^*_k|} \leq \mbox{max}_{f_k \in [0,1/2]} \frac{1-f_k}{f_k^2 + (1-f_k)^2} \le 1.208
    \end{equation*}
Since the above bound holds in each dimension of the hypercube, we can combine them as:
    \begin{equation*}
        \mathbb{E}_{a \in \mathcal{H}^*}\left[\mbox{Distortion}(a)\right] = \frac{\sum_{k=1}^D \mathbb{E}_{a_k \in \mathcal{H}^*_k}\left[ \sum_{u \in \mathcal{N}} |p_{u,k} - a_k|\right]}{\sum_{k=1}^D \sum_{u \in \mathcal{N}} |p_{u,k} - a^*_k|} \leq 1.208
    \end{equation*}
 
 \paragraph{Convergence Rate.}   Now that we have bounded the Distortion of the stationary distribution, we need to consider the convergence rate.  We will not bound the mixing time of the overall Markov chain.  Rather, note that in the preceding analysis, we only used the \textit{marginal} probabilities $\mathcal{H}_k^*$ for every dimension $k$.  Furthermore, the Markov chain defined by sequential deliberation need not walk along edges on $Q$, so we can consider separately the convergence of the chain to the stationary marginal in each dimension.  

After $t$ steps, let $P_{kt} = \Pr[ o_k^t = 1]$ and let $\mathcal{H}^t_k$ denote this distribution.  Assume $f_k \in [0,1/2]$. If the total variation distance\footnote{total variation distance is particularly simple for these distributions with support of just two points: $TVD(\mathcal{H}^t_k, \mathcal{H}^*_k) = |P_{kt} - P^*_k|$} between $\mathcal{H}^t_k$ and $\mathcal{H}^*_k$ is $\epsilon\frac{f_k^2}{f_k^2 + (1-f_k)^2}$, then it is easy to check that the expected total distance to $\mathcal{H}^t_k$ is within a $(1+\epsilon)$ factor of the expected distance to $\mathcal{H}^*_k$, which implies a Distortion of at most $1.208 + \epsilon$ in that dimension. We therefore bound how many steps it takes to achieve total variation distance $\epsilon\frac{f_k^2}{f_k^2 + (1-f_k)^2}$ in any dimension $k$; if this bound holds uniformly for all dimensions $k$, this would imply the overall Distortion is at most $1.208 + \epsilon$, completing the proof.

 For any dimension $k$, two executions of the $2$-state Markov chain along that dimension couple if the agents picked in a time step have the same value in that dimension. At any step, this happens with probability at least $\left(f_k^2 +  (1 - f_k)^2 \right)$. Therefore, the probability that the chains have not coupled in $t$ steps is at most
 $$ \left( 1 - \left(f_k^2 +  (1 - f_k)^2 \right)\right)^t = \left( 2f_k(1-f_k) \right)^t $$
 We therefore need $T$ large enough so that
\begin{eqnarray*}
  \left( 2f_k(1-f_k) \right)^T  & \le & \epsilon\frac{f_k^2}{f_k^2 + (1-f_k)^2}  \\
 \Rightarrow \ T & = & \max_{f_k \in [0,1/2]} \left( \frac{\log \frac{1}{\epsilon}}{\log \frac{1}{2f_k(1-f_k)}} + \frac{\log \frac{1}{f_k^2} + \log(f_k^2 + (1-f_k)^2)}{ \log \frac{1}{2f_k(1-f_k)}} \right) \\
 \Rightarrow \  T & \le & \log_2 \frac{1}{\epsilon} + 2.575
\end{eqnarray*}
 Since this bound of $T$ holds uniformly for all dimensions, this directly implies the theorem. 
\end{proof}



\subsection{Lower Bounds on Distortion}
\label{sec:lb}
We will now show that the Distortion bounds of sequential deliberation are significant, meaning that mechanisms from simpler classes are constrained to have higher Distortion values.  We present a sequence of lower bounds for social choice mechanisms that use increasingly rich information about the space of alternatives on a  median graph $G = (\S,E)$ with a set of agents $\mathcal{N}$ with bliss points $V_{\mathcal{N}} \subseteq \S$. 

We first consider mechanisms that are constrained to choose outcomes in $V_{\mathcal{N}}$. For instance, this captures the random dictatorship that chooses the bliss point of a random agent as the final outcome. It shows that the compromise alternatives found by deliberation do play a role in reducing Distortion. 

\begin{lemma}
\label{lb:random}
Any mechanism constrained to choose outcomes in $V_{\mathcal{N}}$ has Distortion at least 2.  
\end{lemma}
\begin{proof}
    It is easy to see that the $k$-star graph (the graph with a central vertex connected to $k$ other vertices none of which have edges between themselves) is a median graph.  Consider an $|\mathcal{N}|$-star graph where $V_{\mathcal{N}}$ are the non central vertices; that is, each and every agent has a unique bliss point on the periphery of the star.  Then any mechanism constrained to choose outcomes in $V_{\mathcal{N}}$ must choose one of these vertices on the periphery of the star.  The social cost of such a point is $(|\mathcal{N}|-1) \times 2$, whereas the social cost of the optimal central vertex is clearly just $|\mathcal{N}|$.  The Distortion goes to $2$ as $|\mathcal{N}|$ grows large.
    \end{proof}

We draw more contrasts between sequential deliberation and random dictatorship in appendix~\ref{sec:random}. In particular, we show that sequential deliberation dominates random dictatorship on every instance for median graphs, and converges to a Distortion of one for nearly unanimous instances, unlike random dictatorship.  We next consider mechanisms that are restricted to choosing the median of the bliss points of some three agents in $\N$. In particular, this captures sequential deliberation run for $T=1$ steps, as well as mechanisms that generalize dictatorship to an oligarchy composed of at most $3$ agents. This shows that iteratively refining the bargaining outcome has better Distortion than performing only one iteration. 


\begin{lemma}
\label{lemma:oligarch}
Any mechanism constrained to choose outcomes in $V_{\mathcal{N}}$ or a median of three points in $V_{\mathcal{N}}$ must have Distortion at least 1.316.   
\end{lemma}
\begin{proof}
Let $G$ be a median graph; in particular let $G$ be the $D$-dimensional hypercube $\{0,1\}^D$.  For every dimension, an agent has a 1 in that dimension of their bliss point independently with probability $p$. In expectation $p |\mathcal{N}|$ agents' bliss points have a 1 in any given dimension.  We assume  $0 < p < 1/2$ is an absolute constant.  For $a^*$ being the all $0$'s vector,  $\mathbb{E}[SC(a^*)] = D |\mathcal{N}| p$,    where  the randomness  is in the construction.  Now, suppose $D = \mbox{polylog}(|\mathcal{N}|)$.  Then, for any $\beta < 1$, with probability at least $1 - \frac{1}{\mbox{poly}(|\N|)}$ every three points in $V_{\mathcal{N}}$ has at least $\beta (3p^2 - 2p^3) D$ ones. By union bounds, for some $\alpha \in ( \beta,1)$, the social cost of any median of three points in $V_{\mathcal{N}}$ is at least:
\begin{equation}
    \begin{split}
        \mathbb{E}[SC(a)]  & \geq D |\mathcal{N}| \alpha \left[  (3p^2 - 2p^3) (1-p) + p (1-3p^2+2p^3)  \right] \\
        & = D |\mathcal{N}| \alpha \left( 4p^4 - 8p^3 +3p^2 + p \right)
    \end{split}
\end{equation}  
    where again, the randomness is in the construction.  Then there is nonzero probability that
  \begin{equation}
        \label{eq:Distortion}
           \frac{ \mathbb{E}[SC(a)]}{\mathbb{E}[SC(a^*)]}  \geq \alpha \left( 4p^3 - 8p^2 +3p + 1 \right)
        \end{equation} 
    If we choose the argmax of Equation~\ref{eq:Distortion}, we get nonzero probability over the construction that the Distortion is at least $(1.316) \alpha$.  Letting $\alpha$ grow close to $1$ and noting that the nonzero probability over the construction implies the existence of one such instance completes the argument.    
\end{proof}

We finally consider a class of mechanisms that includes sequential deliberation as a special case. We show that any mechanism in this class cannot have Distortion arbitrarily close to $1$. This also shows that sequential deliberation is close to best possible in this class. 


\begin{lemma}
\label{lemma:sp}
Any mechanism constrained to choose outcomes on shortest paths between pairs of outcomes in $V_{\mathcal{N}}$ must have Distortion at least $9/8 = 1.125$.
\end{lemma}
\begin{proof}
    The construction of the lower bound essentially mimics that of Lemma~\ref{lemma:oligarch}.  In this case however, we get that each point on a shortest path between two agents has at least $\alpha (p^2 \times (1-p) + (1-p^2) \times p) D$ ones, so 
$$  \mathbb{E}[SC(a)]  \geq D |\mathcal{N}| \alpha \left[ p^2 \times (1-p) + (1-p^2) \times p \right] =  D |\mathcal{N}| \alpha \left( -2p^3 + p^2 + p \right) 
$$
    So there is nonzero probability over the construction that
        \begin{equation*}
               \frac{ \mathbb{E}[SC(a)]}{\mathbb{E}[SC(a^*)]}  \geq \alpha \left( -2p^2 + p + 1 \right)
        \end{equation*} 
 where the maximum of $-2p^2 + p + 1$ over $p < 1/2$ is $9/8$ when $p = 1/4$.  The rest of the argument follows as in Lemma~\ref{lemma:oligarch}.
\end{proof}

The significance of the lower bound in Lemma~\ref{lemma:sp} should be emphasized: though there is always a Condorcet winner in median graphs, it need not be any agent's bliss point, nor does it need to be Pareto optimal for any pair of agents.  The somewhat surprising implication is that any local mechanism (in the sense that the mechanism chooses locally Pareto optimal points) is constrained away from finding the Condorcet winner.

%% file: properties.tex
\section{Properties of Sequential Deliberation}
\label{sec:properties}
In this section, we study some natural desirable properties for our mechanism: uniqueness of the stationary distribution of the Markov chain, ex-post Pareto-efficiency of the final outcome, and subgame perfect Nash equilibrium.

\paragraph{Uniqueness of the Stationary Distribution}
We first show that the Markov chain corresponding to sequential deliberation converges to a unique stationary distribution on the actual median graph, rather than just showing that the marginals and thus the expected distances from the perspectives of the agents converge.  


To prove uniqueness, it will be helpful to note that the Markov chain defined by sequential bargaining on $G$ by $\mathcal{N}$ only puts nonzero probability mass on points in the median closure $M(V_{\mathcal{N}})$ of $V_{\mathcal{N}}$ (see Definition~\ref{def:medianClosure} and Figure~\ref{figure:medianClosure} for an example).  This is the state space of the Markov chain, and there is a directed edge (i.e., nonzero transition probability) from $x$ to $v$ if there exist $u, w \in V_{\mathcal{N}}$ such that $v = M(x, u, w)$ (where $M(x, u, w)$ is the median of $x, u, w$ by a slight abuse of notation). 

\begin{definition}
\label{def:medianClosure}
    Let $V_{\mathcal{N}} \subseteq \S$ be the set of bliss points of agents in $\mathcal{N}$.  A point $v$ is in $M(V_{\mathcal{N}})$ if $v \in V_{\mathcal{N}}$ or if there exists some sequence $(x^1, y^1), (x^2, y^2), \hdots, (x^t, y^t)$ such that every point in every pair in the sequence is in $V_{\mathcal{N}}$ and there is some $z \in V_{\mathcal{N}}$ s.t.
    \[v = M(x^t, y^t, M(x^{t-1}, y^{t-1}, M( \hdots M(x^1, y^1, z) \hdots ))) \footnote{To interpret this sequence, note that $z$ represents the initial disagreement alternative drawn as the bliss point of a random agent. $(x^t, y^t)$ are the bliss points of the random agents drawn to bargain in every round.  $v$ therefore represents a feasible outcome of sequential deliberation.}\]
\end{definition}
 
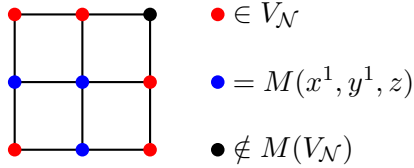
\begin{figure}[!h]
\centering
\begin{tikzpicture}[scale=0.9]

\draw[black, thick] (0,0) -- (0,2) -- (2,2) -- (2,0) -- (0,0);
\draw[black, thick] (1,0) -- (1,2);
\draw[black, thick] (0,1) -- (2,1);

\fill[red] (0,0) circle (0.1); \fill[blue] (0,1) circle (0.1); \fill[red] (0,2) circle (0.1);
\fill[blue] (1,0) circle (0.1); \fill[blue] (1,1) circle (0.1); \fill[red] (1,2) circle (0.1);
\fill[red] (2,0) circle (0.1); \fill[red] (2,1) circle (0.1); \fill[black] (2,2) circle (0.1);

\fill[gray] (1,2.5) circle (0.0);

\fill[red] (3,2) circle (0.1);
\fill[black] (3,2) circle (0.0) node[anchor=west] {$\, \in V_{\mathcal{N}}$};
\fill[blue] (3,1) circle (0.1); 
\fill[black] (3,1) circle (0.0) node[anchor=west] {$\, = M(x^1, y^1, z)$};
\fill[black] (3,0) circle (0.1); 
\fill[black] (3,0) circle (0.0) node[anchor=west] {$\, \notin M(V_{\mathcal{N}})$};

\end{tikzpicture}
\caption{The median closure of the red points is given by the red and blue points.}
\label{figure:medianClosure}
\end{figure}

\begin{theorem}
The Markov chain defined in Theorem~\ref{theorem:main} has a unique stationary distribution.
\label{thm:unique}
\end{theorem}
\begin{proof}

Let $G=(\S,E)$ be a  median graph, let $\mathcal{N}$ be a set of agents, and let $V_{\mathcal{N}} \subseteq \S$ be the set of bliss points of the agents in $\mathcal{N}$.  The Markov chain will have a unique stationary distribution if it is aperiodic and irreducible. 
To see that the chain is aperiodic, note that for any state $o^t = v \in M(V_{\mathcal{N}})$ of the Markov chain at time $t$, there is a nonzero probability that $o^{t+1} = v$.  This is obvious if $v \in V_{\mathcal{N}}$, as the agent corresponding to that bliss point might be drawn twice in round $t+1$ (remember, agents are drawn independently with replacement from $\mathcal{N}$).  If instead $v \notin V_{\mathcal{N}}$, we know by definition of $M(V_{\mathcal{N}})$ that there exist $u, w \in V_{\mathcal{N}}$ and $x \in M(V_{\mathcal{N}})$ such that $v = \mbox{median}(u,w,x)$.  But then $v = \mbox{median}(u,w,v)$. Clearly we can write $d(u,v) = d(u,v) + d(v,v)$ and $d(w,v) = d(w,v) + d(v,v)$, then the fact that $v = \mbox{median}(u,w,x)$ implies that $d(u,w) = d(u,v) + d(v,w)$.  Taken together, these equalities imply that $v$ is the median chosen in round $t+1$.  So in either case, there is some probability that $o^{t+1} = v$.  The period of every state is 1, and the chain is aperiodic.    
    
To argue that the chain is irreducible, suppose for a contradiction that there exist $t, v \in M(V_{\mathcal{N}})$ such that there is no path from $t$ to $v$.  Then $v \notin V_{\mathcal{N}}$, since all nodes in $V_{\mathcal{N}}$ clearly have an incoming edge from every other node in $M(V_{\mathcal{N}})$.  Then by definition there exists some sequence $(x^1, y^1), (x^2, y^2), \hdots, (x^t, y^t) \in V_{\mathcal{N}}$ and some $z \in V_{\mathcal{N}}$ such that 
    \[v = M(x^t, y^t, M(x^{t-1}, y^{t-1}, M( \hdots M(x^1, y^1, z) \hdots ))).\]
Since $z \in V_{\mathcal{N}}$, $z$ must have an incoming edge from $t$.  But then there is a path from $t$ to $v$.  This is a contradiction; the chain must be irreducible as well.  Both properties together show that the Markov chain has a unique stationary distribution.
\end{proof}

\paragraph{Pareto-Efficiency}
\label{sec:pe}
The outcome of sequential deliberation is ex-post Pareto-efficient on a median graph. In other words, in any realization of the random process, suppose the final outcome is $o$; then there is no other alternative $a$ such that $d(a,v) \le d(o,v)$ for every $v \in \mathcal{N}$, with at least one inequality being strict.  This is a weak notion of efficiency, but it is not trivial to show; while it is easy to see that a one shot bargaining mechanism using only bliss points is Pareto efficient by virtue of the Pareto efficiency of bargaining, sequential deliberation defines a potentially complicated Markov chain for which many of the outcomes need not be bliss points themselves.  


\begin{theorem}
\label{thm:PE}
    Sequential deliberation among a set $\mathcal{N}$ of agents, where the decision space $\mathcal{S}$ is a median graph and the initial disagreement point $a^1$ is the bliss point of some agent, yields an ex-post Pareto Efficient alternative.
\end{theorem}
\begin{proof}
Let $G=(\S,E)$ be a  median graph, let $\mathcal{N}$ be a set of agents, and let $V_{\mathcal{N}} \subseteq \S$ be the set of bliss points of the agents in $\mathcal{N}$. It follows from the proof of Corollary~\ref{lemma:hypercube} that without loss of generality we can suppose $\S$ is a hypercube embedding. Consider some realization $(x^1,y^1), (x^2,y^2), ..., (x^T, y^T)$ of sequential bargaining, where $(x^t,y^t) \in \S \times \S$ are the bliss points of the agents drawn to bargain in step $t$.  Let $o^T$ denote the final outcome.  For the sake of contradiction assume there is an alternative $z$ that Pareto-dominates $o^T$, i.e., $d(z,v) \le d(o^T,v)$ for each $v \in \mathcal{N}$, with at least one inequality being strict.

\vspace{-0.75cm}
\begin{figure}[htbp]
    \centering
    \includegraphics[width=0.5\textwidth]{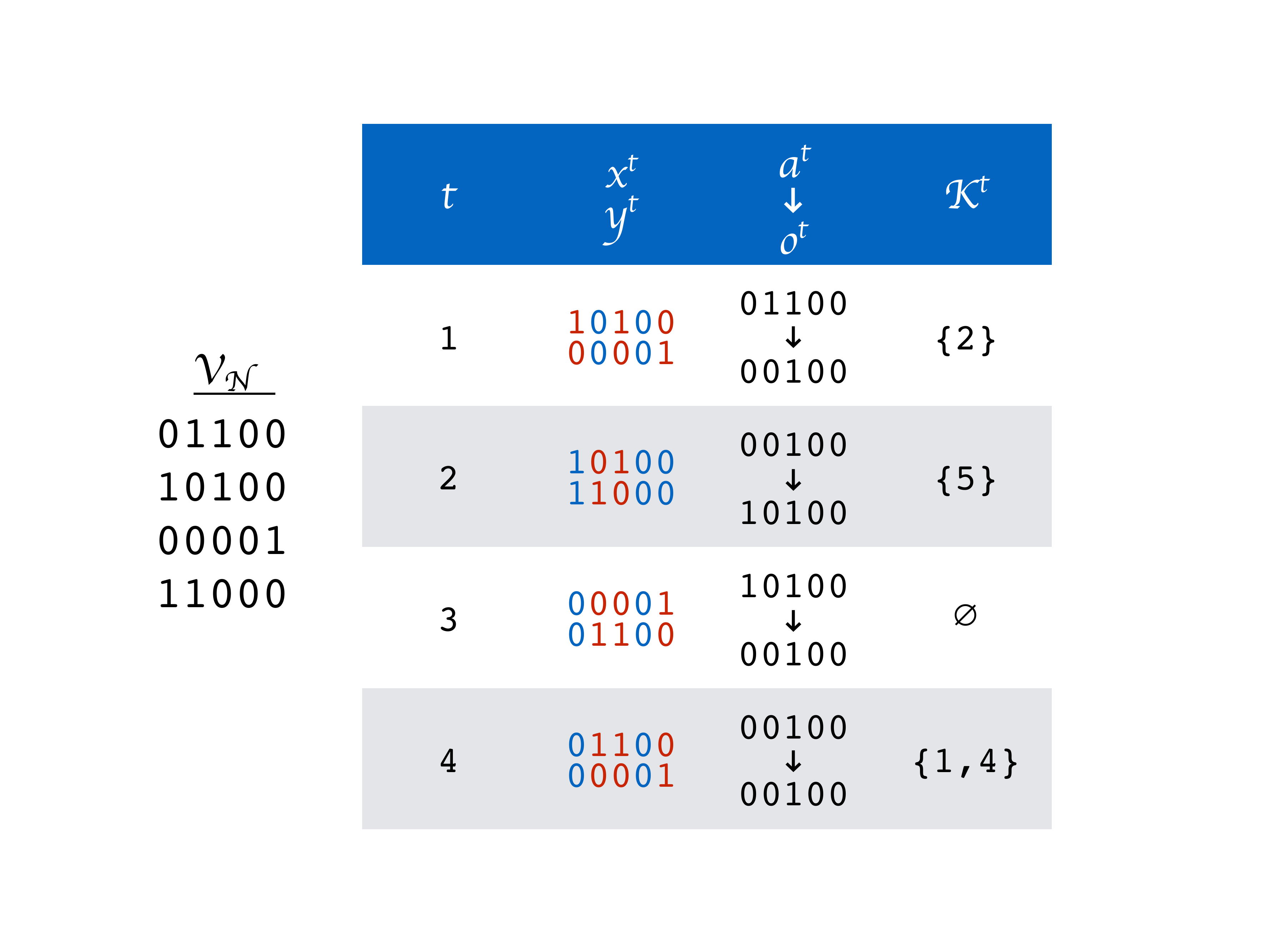}
\vspace{-1cm}
\caption{An example of sequential deliberation with $\mathcal{K}^t$ labeled. The dimensions are numbered $1,2,\ldots, 5$.}
    \label{figure:PE}
\end{figure}

Recall that on the hypercube, the median of three points has the particularly simple form of the dimension wise majority.  Let $\mathcal{K}^t$ be the set of dimensions of the hypercube $\S$ that are ``decided'' by the agents in round $t$ in the sense that these agents agree on that dimension (and can thus ignore the outside alternative in that dimension) and all future agents disagree on that dimension (and thus keep the value decided by bargaining in round $t$).  Formally, $\mathcal{K}^t = \left\{k: \, x^t_k = y^t_k \mbox{ and } \forall t' > t, x^{t'}_k \neq y^{t'}_k \right\}$.  Then by the majority property of the median on the hypercube, for any dimension $k$ such that $k \in \mathcal{K}^t$ for some $t \in \{1, \hdots, T\}$, it must be that $o^T_k = x^t_k$.  An example is shown in Figure~\ref{figure:PE}.  

Consider the final round $T$.  It must be that $\forall k \in \mathcal{K}^T, z_k = o^T_k$.  If this were not the case, $z$ would not be pairwise efficient (i.e., on a shortest path from $x^T$ to $y^T$), whereas $o^T$ is pairwise efficient by definition of the median, so one of the agents in round $T$ would strictly prefer $o^T$ to $z$, violating the dominance of $z$ over $o^T$. 

Next, consider round $T-1$.  Partition the dimensions of $\S$ into $\mathcal{K}^T, \mathcal{K}^{T-1}$ and all others.  Suppose for a contradiction that $\exists k \in \mathcal{K}^{T-1}$ such that $z_k \neq o^T_k$, where $x^{T-1}_k = y^{T-1} = o^T_k$ by definition.  Then the agents in round $T-1$ must strictly prefer $o^T$ to $z$ on the dimensions in $\mathcal{K}^{T-1}$.  But for $k \in \mathcal{K}^T$, we know that $z_k = o^T_k$, so the agents are indifferent between $z$ and $o^T$ on the dimensions in $\mathcal{K}^T$.  Furthermore, for $k \notin \mathcal{K}^T \cup \mathcal{K}^{T-1}, x^{T-1}_k \neq y^{T-1}_k$, so at least one of the two agents at least weakly prefers $o^T$ to $z$ on the remaining dimensions.  But then at least one agent must strictly prefer $o^T$ to $z$, contradicting the dominance of $z$ over $o^T$.

Repeating this argument yields that for all $k \in \mathcal{K}^1 \cup \mathcal{K}^2 \cup \hdots \cup \mathcal{K}^T$, $z_k = o^T_k$.  For all other dimensions, $o^T_k$ takes on the value $a^1$, which is the bliss point of some agent.  Since that agent must weakly prefer $z$ to $o^T$, $z$ must also take the value of her bliss point on these remaining dimensions.  But then $z = o^T$, so $z$ does not Pareto dominate $o^T$, a contradiction.
\end{proof}

\paragraph{Truthfulness of Extensive Forms} 
\label{sec:truthful}
Finally, we show that sequential deliberation has truth-telling as a sub-game perfect Nash equilibrium in its induced extensive form game.  Towards this end, we formalize a given round of bargaining as a 2-person non-cooperative game between two players who can choose as a strategy to report any point $v$ on a median graph; the resulting outcome is the median of the two strategy points chosen by the players and the disagreement alternative presented.  The payoffs to the players are just the utilities already defined; i.e., the player wishes to minimize the distance from their true bliss point to the outcome point.  Call this game the non-cooperative bargaining game (NCBG).


The extensive form game tree  defined by non-cooperative bargaining consists of $2T$ alternating levels: Nature draws two agents at random, then the two agents play NCBG and the outcome becomes the disagreement alternative for the next NCBG.  The leaves of the tree are  a set of points in the median graph; agents want to minimize their expected distance to the final outcome.

\begin{theorem}
\label{thm:SPNE}
Sequential NCBG on a median graph has a sub-game perfect Nash equilibrium where every agent truthfully reports their bliss point at all rounds of bargaining. 
\end{theorem}

\begin{proof}
The proof is by backward induction. Let $G=(\S,E)$ be a median graph. In the base case, consider the final round of bargaining between agents $u$ and $v$ with bliss points $p_u$ and $p_v$ and disagreement alternative $a$. The claim is that $u$ playing $p_u$ and $v$ playing $p_v$ is a Nash equilibrium. By Lemma~\ref{lemma:embed}, we can embed $G$ isometrically into a hypercube $Q$ as $\phi:G \rightarrow Q$ and consider the bargaining on this embedding. Then for any point $z$ that agent $u$ plays
\[ d(p_u, M(z,p_v,a)) = d(\phi(p_u), M(\phi(z), \phi(p_v), \phi(a))) \]
The median on the hypercube is just the bitwise majority, so if $u$ plays some $z$ where for some dimension $\phi_k(z) \neq \phi_k(p_u)$, it can only increase $u$'s distance to the median.  So playing $p_u$ is a best response.

For the inductive step, suppose $u$ is at an arbitrary subgame in the game tree with $t$ rounds left, including the current bargain in which $u$ must report a point, and assuming truthful play in all subsequent rounds. Let $\{(x^1,y^1), (x^2,y^2), ..., (x^t, y^t)\}$ represent $(x^1, y^1)$ as the outside alternative and other agent bliss point against which $u$ must bargain, $(x^2, y^2)$ as the bliss points of the agents drawn in the next round, and so on.  We want to show that it is a best response for agent $u$ to choose $p_u$, i.e., to truthfully represent her bliss point. Define
    \[M_u^t := M(x^t, y^t, M(x^{t-1}, y^{t-1}, M( \hdots M(x^1, y^1, p_u) \hdots )))\]
    where $M(\cdot)$ indicates the median, guaranteed to exist and be unique on the median graph.  Also, for any point $z$, similarly define
    \[ M_z^t := M(x^t, y^t, M(x^{t-1}, y^{t-1}, M( \hdots M(x^1, y^1, z) \hdots ))) \]

Suppose by contradiction that $p_u$ is not a best response for agent $u$, then there must exist $z \neq p_u$ and some $r > 0$ such that $d(p_u, M_u^r) > d(p_u, M_z^r)$. We embed $G$ isometrically into a hypercube $Q$ as $\phi:G \rightarrow Q$. Then by the isometry property, $d(\phi(p_u), \phi(M_u^r)) > d(\phi(p_u), \phi(M_z^r))$.  By the proof of Corollary~\ref{lemma:hypercube}, we can pretend the process occurs on the hypercube.

Consider some dimension $k$. If $\phi_k(x^t) = \phi_k(y^t)$ for some $t \le r$, then this point becomes the median in that dimension, so the median becomes independent of $\phi_k(w)$, where $w$ is the initial report of agent $u$. Till that time, the bargaining outcome in that dimension is the same as $\phi_k(w)$. In either case, for all times $t \le r$ and all $k$, we have:
$$ |\phi_k(p_u) - \phi_k(M_u^t) |  \le  |\phi_k(p_u) -  \phi_k(M_z^t)| $$ 
Summing this up over all dimensions, $d(\phi(p_u), \phi(M_u^t)) \le d(\phi(p_u), \phi(M_z^t))$, which is a contradiction. Therefore, $p_u$ was a best response for agent $u$\footnote{It is important to note that we are \textit{not} assuming that agent $u$ will not bargain again in the subgame; there are no restrictions on the values of $\{(x^1,y^1), (x^2,y^2), ..., (x^t, y^t)\}$.}. Therefore, every agent truthfully reporting their bliss points at all rounds is a subgame perfect Nash equilibrium of Sequential NCBG.
\end{proof}


%% file: generalMetrics.tex
\section{General Metric Spaces}
\label{sec:general}
We now work in the very general setting that the set $\S$ of alternatives are points in a finite metric space equipped with a distance function $d(\cdot)$ that is a metric. As before, we assume each agent $u \in \N$ has a bliss point $p_u \in \S$.  An agent's disutility for an alternative $a \in \S$ is simply $d(p_u,a)$. We first present results for the Distortion, and subsequently define the second moment, or Squared-Distortion. For both measures, we show that the upper bound for sequential deliberation is at most a constant regardless of the metric space. 


\begin{theorem}
\label{theorem:Distortion-general}
The Distortion of sequential deliberation is at most 3 when the space of alternatives and bliss points lies in some metric, and this bound is tight.
\end{theorem}
\begin{proof}
Each agent $u \in \N$ has a bliss point $p_u \in \S$.  An agent's disutility for an alternative $a \in \S$ is simply $d(p_u,a)$.  Let $a^* \in \S$ be the social cost minimizer, i.e., the generalized median.  For convenience, let $Z_a = d(a, a^*)$ for $a \in \S$.  By a slight abuse of notation, let $Z_i = d(p_i, a^*)$ for $i \in \N$, i.e., the distance from agent $i$'s bliss point to $a^*$.

We will only use the assumption that $\B(u, v, a)$ finds a Pareto efficient point for $u$ and $v$, so rather than taking an expectation over the choice of the disagreement alternative, we take the worst case.  Let $a^*$ be the social cost minimizer with social cost $OPT$. We can write then write expected worst case social cost of a step of deliberation as:
\begin{equation*}
\begin{split}
    & \sum_{i \in \N} \sum_{j \in \N} \frac{1}{|\N|^2} \max_{a \in \S} \sum_{k \in \N} d(\B(i, j, a), p_k) \\
    & \leq \sum_{i,j \in \N} \frac{1}{|\N|^2} \max_{a \in \S} \sum_{k \in \N} d(\B(i, j, a), a^*) + Z_k = OPT + \sum_{i,j \in \N} \frac{1}{|\N|} \max_{a \in \S} \left( d(\B(i, j, a), a^*) \right) \mbox{ \big{[}Triangle inequality\big{]}}\\
    & \leq OPT + \sum_{i,j \in \N} \frac{1}{2|\N|} \max_{a \in \S} \left( d(\B(i, j, a), p_i) + Z_i + d(\B(i, j, a), p_j) + Z_j \right) \mbox{ \big{[}Triangle inequality\big{]}} \\
    & \leq OPT + \sum_{i,j \in \N} \frac{1}{2|\N|} \left(2Z_i + 2Z_j\right) \mbox{ \big{[}Pareto efficiency of $\B(\cdot)$\big{]}} \\
    & = 3 OPT
\end{split}
\end{equation*} 
Since this holds for the worst case choice of disagreement alternative, it holds over the whole sequential process.  The tight example is a weighted graph that can be constructed as follows: start with a star with a single agent's bliss point located at each leaf of the star (i.e., an $|N|$-star) where every edge in the star has weight (or distance) 1.  Now, for every pair of agents, add a vertex connected to the bliss points of both of the agents with weight $1-\epsilon$ for $\epsilon > 0$.  Then as $|N| \rightarrow \infty$, every round of Nash bargaining selects one of these pairwise vertices.  But the central vertex of the star has social cost $|N|$ whereas the pairwise vertices all have social cost approaching $3|N|$ as $\epsilon \rightarrow 0$.
\end{proof}

The bound of $3$ above is quite pessimistic. The metric space employed in the lower bound is contrived in the following sense: Every pair of agents has some unique (to that pair) alternative they very slightly prefer to the social optimum. For structured spaces, we expect the bound to be much better. We have already shown this for median spaces. In appendix~\ref{sec:budget}, we provide more evidence in this direction by considering a structured space motivated by budgeting applications that is not a median graph. For this space, we show that sequential deliberation has Distortion at most $4/3$.

\subsection{Second Moment of Social Cost}
\label{sec:second}
We now show that for any metric space, sequential deliberation has a crucial advantage in terms of the distribution of outcomes it produces.  For this, we consider the second moment, or the expected squared social cost.  Recall that the {\em social cost} of an alternative $a \in \S$ is given by $ SC(a) = \sum_{u \in \N} d(p_u,a)$. Let $a^* \in \S$ be the minimizer of social cost, {\em i.e.}, the {\em generalized median}. Then define:
 $$ \mbox{Squared-Distortion} = \frac{\mathbb{E}[(SC(a))^2]}{(SC(a^*))^2}$$
 where the expectation is over the set of outcomes $a$ produced by Sequential Deliberation.\footnote{The motivation for considering Squared-Distortion instead of the standard deviation is that the latter might prefer a more deterministic mechanism with a worse social cost, a problem that the Squared-Distortion avoids.}
  
We will show that sequential deliberation has Squared-Distortion upper bounded by a constant. This means the standard deviation in social cost of the distribution of outcomes is comparable to the optimal social cost. This has a practical implication: A policy designer can run sequential deliberation for a few steps, and be sure that the probability of observing an outcome that has $\gamma$ times the optimal social cost is at most $O(1/\gamma^2)$.   In contrast, we show that Random Dictatorship (choosing an agent uniformly at random and using her bliss point as the solution) has unbounded Squared-Distortion, which means its standard deviation in social cost cannot be bounded.  A different way of viewing our result is that deliberation between agents eliminates the outlier agent, and concentrates probability mass on central outcomes.

\medskip
For convenience, let $Z_a = d(a, a^*)$ for $a \in \S$.  By a slight abuse of notation, let $Z_i = d(p_i, a^*)$ for $i \in \N$, i.e., the distance from agent $i$'s bliss point to $a^*$. We will need the following technical lemma bounding these distances.  The lemma addresses the following question: given an arbitrary agent $u$, how far away can the outcome of a bargaining round with agents $i$ and $j$ and disagreement alternative $a$ be from $p_u$?  The answer is that it cannot be much further than the values of $Z_u$ and the smaller of $Z_i, Z_j, Z_a$.  The two min functions in the bound serve to eliminate outliers, and this is crucial for bounding the Squared-Distortion.  

\begin{lemma} For all $i,j,u \in \N$ and $a \in \S$ we have that
\label{lemma:bargain_to_u}
$$ d(\B(i, j, a), p_u) \le  Z_u + 2 \min{(Z_i,Z_j)} + \min{(Z_a, \max{(Z_i,Z_j)})} $$
\end{lemma}
\begin{proof}
Assume w.l.o.g. that $Z_i \leq Z_j$.  Then the lemma statement reduces to 
\begin{equation}
\label{equation:lemma_statement}
    d(\B(i, j, a), p_u) \le  Z_u + 2 Z_i + \min{(Z_a, Z_j)}
\end{equation}
Recall that Nash bargaining asserts that given a disagreement alternative $a$, agents $i$ and $j$ choose that alternative $o\in \S$ that maximizes:
$$ \mbox{Nash product } = \left( d(p_i,a) - d(p_i,o) \right) \times   \left( d(p_j,a) - d(p_j,o) \right) $$
Maximizing this on a general metric yields that $d(\B(i, j, a)$ will be chosen on the $p_i$ to $p_j$ shortest path (that is, the Pareto efficient frontier) at a distance of $\frac{d(p_i,p_j)}{2} + \frac{d(p_i,a) - d(p_j,a)}{2}$ from $p_i$.  Therefore, we have that:
\begin{equation}
\label{equation:bargain_to_i}
d(\B(i,j,a), p_i) = \frac{d(p_i,p_j)}{2} + \frac{d(p_i,a) - d(p_j,a)}{2}
\end{equation} 
Now we can show equation~\ref{equation:lemma_statement} by repeatedly using the triangle inequality. We have
\begin{equation*}
\begin{split}
    d(\B(i, j, a), p_u) & \le d(p_i,p_u) + d(\B(i,j,a), p_i) \\
    & = d(p_i,p_u) + \frac{d(p_i,p_j)}{2} + \frac{d(p_i,a) - d(p_j,a)}{2} \mbox{ \big{[}equation~\ref{equation:bargain_to_i}\big{]}} \\
    & \leq d(p_i, a^*) + d(p_u, a^*) + \frac{d(p_i, a^*) + d(p_j, a^*) + d(p_i, a^*) + d(a, a^*)}{2} - \frac{d(p_j,a)}{2} \\ 
    & = Z_u + 2Z_i + \frac{Z_j + Z_a - d(p_j, a)}{2} \\
\end{split}
\end{equation*}
Note that we can apply the triangle inequality again to get two bounds, since $Z_j \leq d(p_j,a) + Z_a$ and $Z_a \leq d(p_j,a) + Z_j$.  Since both bounds must hold, we have that $d(\B(i, j, a), p_u) \leq Z_u + 2Z_i + \min{(Z_a,Z_j)}$.
\end{proof}

It is not hard to see that this bound is tight in the worst case.  Suppose agents $p_i, p_j, p_u$ form the leaves of a 3-star (with weight 1 edges) and $a$ is the center of the star.  Now suppose there is a point $o$ connected to $p_i$ and to $p_j$ each by edges of weight $1-\epsilon$ for $\epsilon > 0$. $o$ will clearly be the bargaining outcome, so $d(\B(i, j, a), p_u) = 3 - \epsilon$.  But the right hand side of the inequality, $Z_u + 2 \min{(Z_i,Z_j)} + \min{(Z_a, \max{(Z_i,Z_j)})} = 3$, so the inequality becomes tight as $\epsilon \rightarrow 0$. 

\medskip 
Using this characterization, we show that the Squared-Distortion bound remains constant even for general metrics. 


\begin{theorem}
The Squared-Distortion of sequential deliberation for $T \geq 1$ is at most $41$ when the space of alternatives and bliss points lies in some metric. Furthermore, the Squared-Distortion of random dictatorship is unbounded.
\label{thm:squaredDistortionGeneral}
\end{theorem}
\begin{proof}
Let $OPT$ be the squared social cost of $a^*$.  Note that OPT does not depend on $T$.  In particular, $OPT = \left( \sum_{i \in \N} Z_i \right)^2 = \sum_{i,j \in \N} Z_i Z_j$.

For arbitrary $T$, the disagreement alternative in the final step of bargaining generated by sequential deliberation.  We will therefore rely on the first moment bound given in Theorem~\ref{theorem:Distortion-general} for such an outcome. Let $ALG$ be the expected squared social cost of sequential deliberation with $T$ steps, where the disagreement alternative $a \in \S$ is used in the final round of bargaining.  $i$ and $j$ are the last two agents to bargain, and $p_k$ is the bliss point of an arbitrary agent.  We can write $ALG$ as
\begin{equation*}
\begin{split}
    ALG & = \sum_{i \in \N} \sum_{j \in \N} \frac{1}{|\N|^2} \left(\sum_{k \in \N} d(\B(i, j, a), p_k)\right)^2 \\
    & \leq \frac{1}{|\N|^2} \sum_{i,j \in \N} \left(\sum_{k \in \N} Z_k + 2\min{(Z_i,Z_j)} + \min{(Z_a, \max{(Z_i,Z_j)})}\right)^2 \mbox{ \big{[}By lemma~\ref{lemma:bargain_to_u}\big{]}}\\
    & = \frac{1}{|\N|^2}\sum_{i,j\in \N} \left(2|\N|\min{(Z_i,Z_j)} + |\N|\min{(Z_a, \max{(Z_i,Z_j)})} + \left( \sum_{k \in \N} Z_k \right)  \right)^2 \\
\end{split}
\end{equation*}

Now we expand the square and analyze term by term, using the facts that $\min{(x,y)}^2 \leq x \times y$ and $\max{(x,y)} \leq x+y$.
\begin{equation*}
\begin{split}
    & = \frac{1}{|\N|^2}\sum_{i,j \in \N} \Bigg{(} 4|\N|^2\min{(Z_i,Z_j)}^2 + |\N|^2\min{(Z_a, \max{(Z_i,Z_j)})}^2 + \left( \sum_{k \in \N} Z_k \right)^2 \\
    & \hspace{1.5cm} + 4|\N|^2\min{(Z_i,Z_j)}\min{(Z_a, \max{(Z_i,Z_j)})} + 4|\N|\min{(Z_i,Z_j)}\left( \sum_{k \in \N} Z_k \right) \\
    & \hspace{1.5cm} + 2|\N|\min{(Z_a, \max{(Z_i,Z_j)})}\left( \sum_{k \in \N} Z_k \right) \Bigg{)}\\ 
    & \leq \frac{1}{|\N|^2}\sum_{i,j \in \N} \Bigg{(} 4|\N|^2 Z_i Z_j + |\N|^2\min{(Z_a, \max{(Z_i,Z_j)})}^2 + OPT \\
    & \hspace{1.5cm} + 4|\N|^2 Z_i Z_a + 4|\N| Z_i \left( \sum_{k \in \N} Z_k \right) + 2|\N|Z_a \left( \sum_{k \in \N} Z_k \right) \Bigg{)} \\
\end{split}
\end{equation*}    

We can trivially sum out the terms not involving $a$, leaving us with inequality~\ref{ineq:Za_terms}.
\begin{equation}
\label{ineq:Za_terms}
ALG \leq 9OPT + \sum_{i,j \in \N} \Bigg{(} Z_a (Z_i + Z_j) + 4 Z_i Z_a + \frac{2}{|\N|}Z_a \left( \sum_{k \in \N} Z_k \right) \Bigg{)}
\end{equation}

Note that the triangle inequality implies that for any $i \in \N$, $Z_a \leq d(a,p_i) + Z_i$.  Applying this repeatedly in inequality~\ref{ineq:Za_terms} and simplifying yields:
\begin{equation*}
\begin{split}
    ALG & \leq 9OPT + \sum_{i,j \in \N} \left( 6 Z_i Z_j + 5 d(a,p_j)Z_i + d(a,p_i)Z_j + \frac{2}{|\N|} (d(a,p_i) + Z_i) \left( \sum_{k \in \N} Z_k \right) \right) \\
    & = 17 OPT + \sum_{i \in \N} \left( 2d(a,p_i)\left( \sum_{k \in \N} Z_k \right) + \sum_{j \in \N} 5 d(a,p_j)Z_i + d(a,p_i)Z_j \right)
\end{split}
\end{equation*}

Recall that Theorem~\ref{theorem:Distortion-general} implies that $\sum_{i \in \N} d(a, p_i) \leq 3 \sum_{i \in \N} Z_i$.  
\begin{equation*}
\begin{split}
    ALG & \leq 17 OPT + 6 OPT + 15 OPT + 3 OPT = 41 OPT
\end{split}
\end{equation*}

\medskip
It is easy to see that Random Dictatorship has an unbounded Squared-Distortion.  Recall that Random Dictatorship chooses the bliss point of an agent chosen uniformly at random. Consider the simple graph with two nodes, a fraction $f$ of the agents on one node and $1-f$ on the other.  Let $f < 1/2$.  Then the expected squared social cost of Random Dictatorship is just $f (1-f)^2 + (1-f) f^2$ whereas the optimal solution has squared social cost $f^2$, so Random Dictatorship has Squared-Distortion $(1-f)^2/f + (1-f)$, which is unbounded as $f \rightarrow 0$.
\end{proof}

Though we omit the proof, we can show using a simpler analysis  that the Squared Distortion of sequential deliberation for median spaces is at most a factor of $3$.

%% file: open.tex
\section{Open Questions}
Our work is the first step to developing a theory around practical deliberation schemes. We suggest several future directions. First, we do not have a general characterization of the Distortion of sequential deliberation for metric spaces.  We have shown that for general metric spaces there is a small but pessimistic bound on the Distortion of 3, but that for specific metric spaces the Distortion may be much lower.  We do not have a complete characterization of what separates these good and bad regimes.  

More broadly, an interesting question is extending our work to take opinion dynamics into account, {\em i.e.}, proving stronger guarantees if we assume that when two agents deliberate, each agent's opinion moves slightly towards the other agent's opinion and the outside alternative.  Furthermore, though we have shown that all agents deliberating at the same time does not improve on dictatorship, it is not clear how to extend our results to more than two agents negotiating at the same time. This runs into the challenges in understanding and modeling multiplayer bargaining~\cite{Harsanyi,Krishna,OddManOut}.

Finally, it would be interesting to conduct experiments to measure the efficacy of our framework on complex, real world social choice scenarios. There are several practical hurdles that need to be overcome before such a system can be feasibly deployed. In a related sense, it would be interesting to develop an axiomatic theory for deliberation, much like that for bargaining~\cite{NashBargaining}, and show that sequential deliberation arises naturally from a set of axioms. 


%% file: appendix.tex
\section{Random Dictatorship and $N$-Person Bargaining on Median Graphs}
\label{sec:random}
A simple baseline algorithm for our problem is Random Dictatorship: Output the bliss point $p_u$ of a random agent $u \in \N$. Such an algorithm also satisfies the desiderata that we presented in Section~\ref{sec:properties}: The outcome is trivially Pareto-efficient and the mechanism is truthful. The Distortion of Random Dictatorship on social cost is at most $2$, and Theorem~\ref{lb:random} shows this bound is tight. We now show two ways that sequential deliberation is superior to Random Dictatorship despite their seeming similarity.   These statements are  stronger than what the worst-case Distortion bounds imply.

\begin{theorem}
\label{thm:random}
For any set of agents $\N$ and outcomes $\S$ that define a median graph, the expected social cost of sequential deliberation is at most that of Random Dictatorship.
\end{theorem}
\begin{proof}
As in the proof of Theorem~\ref{theorem:main}, we consider the hypercube embedding, and decompose the social cost along each dimension $k$. Let $f_k$ denote the fraction of agents mapping to $1$ along dimension $k$. Then, the previous proof showed that the social cost of sequential deliberation is
$$ \mbox{Social Cost of Sequential Deliberation } =  |\mathcal{N}|\left(\frac{f_k(1-f_k)}{1+2f_k^2-2f_k} \right)$$
Random Dictatorship would choose an agent mapping to $1$ with probability $f_k$, in which case a fraction $(1-f_k)$ of agents pay cost $1$, and vice versa. Therefore its social cost is
$$ \mbox{Social Cost of Random Dictatorship } = |\mathcal{N}| \left(2 f_k(1-f_k)\right)$$
It is an easy exercise to show that for all $f_k \in [0,1]$, the former cost is at most the latter cost. Since this holds dimension by dimension, the overall social cost for sequential deliberation is at most that of Random Dictatorship.
\end{proof}

This shows that instance by instance, sequential deliberation does at least as well on the first moment of social cost as Random Dictatorship, and in addition, it has bounded second moment in the worst case, which Random Dictatorship does not have.

A different comparison between sequential deliberation and Random Dictatorship concerns the convergence to optimality under the ``easy case'' of nearly unanimous agents, i.e., when most agents have exactly the same bliss point.  The result for sequential deliberation is a simple corollary of Theorem~\ref{theorem:main} but is presented here for purpose of its comparison with Random Dictatorship.

\begin{definition}
    A set of agents $\N$ is said to be $\epsilon$-\textbf{unanimous} on a set of outcomes $\S$ if all but an $\epsilon$ fraction of the agents have the same bliss point in $\S$.
\end{definition}

\begin{corollary}
\label{cor:unanimous}
    Sequential deliberation with a $\epsilon$-unanimous set of agents on a median graph has Distortion at most  $(1+\epsilon)$. In contrast, Random Dictatorship has Distortion $2 - \epsilon$.
\end{corollary}
\begin{proof}
As in Theorem~\ref{theorem:main}, we can analyze sequential deliberation on the hypercube.  As before, let $f_k$ be the fraction of agents with a 1 in dimension $k$ and assume w.l.o.g. that $f_k \in [0,1/2]$ so that the optimal point is the all 0 vector.  In the proof of Theorem~\ref{theorem:main}, we derived the following as a bound for the Distortion of sequential deliberation. 
    \begin{equation}
        \label{equation:max}
        \mbox{max}_{f_k \in [0,1/2]} \frac{1-f_k}{1+2f_k^2 - 2f_k} \le 1.208
    \end{equation}
If $\epsilon \geq 0.208$ then the result is immediate. Otherwise, note that under $\epsilon$-unanimity, it must in fact be the case that $f_k \in [0, \epsilon]$, so the maximization is over a more restricted domain.  Furthermore, for $\epsilon < 0.208$, the function in Equation~\ref{equation:max} is monotonically increasing and thus obtains its maximum at the right boundary $\epsilon$.  It is easy to see this by taking the derivative.  Thus, for $\epsilon < 0.208$, the expected Distortion of the limiting distribution of sequential deliberation is upper bounded by
    
    \begin{equation*}
        \begin{split}
            \mbox{max}_{f_k \in [0, \epsilon]} \frac{1-f_k}{1+2f_k^2 - 2f_k} & = \frac{1-\epsilon}{1+2\epsilon^2 - 2\epsilon} \\
            & \leq \frac{1-\epsilon}{1+2\epsilon^2 - 2\epsilon - \frac{2\epsilon^3}{1+\epsilon}} \\
            & = 1 + \epsilon
        \end{split}
    \end{equation*}
    
In contrast it is easy to see that even for the simple median graph consisting of two points (the one dimensional hypercube), Random Dictatorship can have Distortion 2 even as the agents are $\epsilon$-unanimous.  On such a graph, the optimal social cost is just $\epsilon |\N|$ but Random Dictatorship has expected social cost $(1-\epsilon) \epsilon |\N| + \epsilon (1-\epsilon) |\N| = 2 \epsilon (1-\epsilon) |\N|$.  The ratio goes to 2 as $\epsilon$ goes to $0$.  This shows Random Dictatorship can actually obtain its worst case Distortion even on the simplest possible instances.
\end{proof}

\paragraph{$N$-Person Nash Bargaining.} Finally, we present a canonical instance of bargaining that devolves into Random Dictatorship by allowing \textit{all} agents to bargain in a single round.  Consider a line metric. Suppose $n_1 > 0$ agents lie at $0$, and $n_2 = N - n_1 > 0$ agents at the point $1$. Suppose the disagreement outcome $o$ is any point $x \in [0,1]$.  Suppose all $N$ agents simultaneously bargain to maximize the Nash product of the increase in their utility with respect to $o$. It is easy to show the following result:

\begin{theorem}
$N$-person Nash bargaining for $N$ agents whose bliss points lie at $0$ and $1$ on a line, with the disagreement outcome $x \in [0,1]$, outputs the same point $x$ as the final outcome.
\end{theorem}

This follows because any deviation makes at least one agent strictly less happy, and this agent will stick to disagreeing. Suppose $x$ corresponds to the bliss point of a random agent, then this mechanism coincides with Random Dictatorship and has Distortion $2$. This shows that the Distortion of $N$-person Nash bargaining even with a random bliss point as disagreement outcome is worse than pairwise deliberation. Note that the result only needs enough agents so that there is one agent at $0$ and the other at $1$.  This confirms (on the line) the intuitive observation that deliberation in large groups tends to break down and produce random outcomes.


\section{Distortion of Deliberation on the Unit Simplex}
\label{sec:budget}
\input{budget.tex}

%% file: budget.tex
We have shown that for general metric spaces there is a small but pessimistic bound on the Distortion of 3, but that for median graphs, the Distortion is much lower.  This begs the question: {\em Are there spaces beyond median graphs for which sequential deliberation has Distortion much less than $3$?}  Though we leave the characterization of such spaces as an open question, we show that our results do extend beyond median graphs. Consider the following special case motivated by budgeting applications. The outcome space $\S$ is the $d$-dimensional standard simplex. Agents are located at the vertices of this simplex, and the distances are $\ell_1$. It is clear that this is not a median space; in fact, the shortest paths between pairs of vertices  from a triplet do not intersect. We can view the vertices as items of unit size, and the restriction to a simplex as a unit budget constraint. Agents' bliss points correspond to single items, while the outcome space is all possible fractional allocations. We show that the Distortion of sequential deliberation in this setting is $4/3$, while the corresponding bound for random dictatorship is again $2$. Though this setting is stylized in the sense that approval voting on items coincides with the social optimum, our scheme itself is general and not tied to the designer knowing that utilities lie in this specific space.

The space $\S$ is the $d$-dimensional standard simplex: 
$$ \S = \left\{ \vec{x} \in \Re^d \ | \  \sum_{j=1}^d x_j = 1; \ \  x_j \ge 0 \ \forall j \in \{1,2,\ldots, d\} \right\}$$
The are $N$ agents whose bliss points are its vertices. We denote the vertex with $x_j = 1$ as $v_j$. Suppose the probability mass of agents located at vertex $j$ is $p_j$, so that $\sum_{j=1}^d p_j = 1$.

Consider a disagreement alternative $\vec{a} \in \S$, and Nash bargaining between agents at $v_i$ and $v_j$ with $j \neq i$. It is easy to check that this produces the outcome with $o_i = (1 + a_i - a_j)/2$, and $o_j =  (1+ a_j - a_i)/2$. When $i = j$, we have $o_i = 1$.

Let $\pi_{\sigma}$ denote the distribution over alternatives $\sigma \in \S$ in steady state. Define the expected value of coordinate $i$ as: 
$$ s_i = \sum_{\sigma} \pi_{\sigma} \sigma_i \qquad \forall i \in \{1,2,\ldots, d\}$$
The steady state conditions imply the set of equations:
$$ \sum_{\sigma} \pi_{\sigma} \sigma_i = p_i^2 +  \sum_{\sigma} \pi_{\sigma} \sum_{j \neq i} 2 p_i p_j \left(\frac{1 + \sigma_i - \sigma_j}{2}\right)$$
The LHS is the expected value of $\sigma_i$ according to distribution $\pi$. The RHS is how the value evolves in one step, which should yield the same quantity. By definition of $s_i$, we therefore have:
$$ s_i = p_i^2 + \sum_{j \neq i} 2 p_i p_j \left( \frac{1 + s_i - s_j}{2} \right) \qquad \forall i \in \{1,2,\ldots, d\}$$
and $ \sum_{i=1}^d s_i = 1$.
Solving for this system, we obtain:
$$ s_i = \frac{p_i}{1-p_i} \times \frac{1}{\sum_{j=1}^d \frac{p_j}{1-p_j}}$$

Note next that the optimal social cost is $OPT = \min_i (1-p_i)$. Without loss of generality, assume this is $1-p_1$. Let $\alpha = p_1$, so that $OPT = 1 - \alpha$. The social cost with respect to the steady state distribution is given by
\begin{eqnarray*}
ALG  & = & \sum_{\sigma} \pi_{\sigma} \sum_i p_i (1-\sigma_i) = \sum_i p_i (1 - s_i) = 1 - \sum_i p_i s_i  \\ 
  & = & 1 - \frac{\sum_i  \frac{p_i^2}{1-p_i}}{\sum_j \frac{p_j}{1-p_j}} = \frac{\alpha + \sum_{j > 1} p_j}{\frac{\alpha}{1-\alpha} + \sum_{j > 1} \frac{p_j}{1-p_j}}  = \frac{1}{\frac{\alpha}{1-\alpha} + \sum_{j > 1} \frac{p_j}{1-p_j}} \\
  & \le &  \frac{1}{\frac{\alpha}{1-\alpha} + \sum_{j > 1} p_j}  =  \frac{1}{\frac{\alpha}{1-\alpha} + 1 - \alpha} = \frac{1-\alpha}{1 - \alpha + \alpha^2}
\end{eqnarray*}
Therefore, the Distortion of sequential deliberation is at most 
$$ \max_{\alpha \in [0,1]} \frac{1}{1-\alpha + \alpha^2} = \frac{4}{3}$$

\medskip
Note that on this instance, the distortion of random dictatorship is 
$$ \frac{\alpha(1-\alpha) + \sum_{i > 1} p_i(1-p_i)}{1-\alpha}  \le 1+ \alpha \le 2$$
and this bound can easily be shown to be tight.

%% file: main.bbl
\begin{thebibliography}{10}
\providecommand{\url}[1]{\texttt{#1}}
\providecommand{\urlprefix}{URL }

\bibitem{IteratedMajority}
Airiau, S., Endriss, U.: Iterated majority voting. In: International Conference
  on Algorithmic Decision Theory. pp. 38--49. Springer (2009)

\bibitem{anshelevich2015approximating}
Anshelevich, E., Bhardwaj, O., Postl, J.: Approximating optimal social choice
  under metric preferences. In: AAAI. vol.~15, pp. 777--783 (2015)

\bibitem{anshelevich2016randomized}
Anshelevich, E., Postl, J.: Randomized social choice functions under metric
  preferences. 25th International Joint Conference on Artificial Intelligence
  (2016)

\bibitem{BandeltB84}
Bandelt, H.J., Barthelemy, J.P.: Medians in median graphs. Discrete Applied
  Mathematics  8(2),  131 -- 142 (1984)

\bibitem{Barbera}
Barbera, S., Gul, F., Stacchetti, E.: Generalized median voter schemes and
  committees. Journal of Economic Theory  61(2),  262 -- 289 (1993)

\bibitem{OddManOut}
Bennett, E., Houba, H.: Odd man out: {B}argaining among three players. Working
  Papers 662, UCLA Department of Economics (May 1992),
  \url{http://www.econ.ucla.edu/workingpapers/wp662.pdf}

\bibitem{BinmoreExperiment}
Binmore, K., Shaked, A., Sutton, J.: Testing noncooperative bargaining theory:
  {A} preliminary study. The American Economic Review  75(5),  1178--1180
  (1985)

\bibitem{Rubinstein2}
Binmore, K., Rubinstein, A., Wolinsky, A.: The nash bargaining solution in
  economic modelling. The RAND Journal of Economics  17(2),  176--188 (1986)

\bibitem{boutilier2015optimal}
Boutilier, C., Caragiannis, I., Haber, S., Lu, T., Procaccia, A.D., Sheffet,
  O.: Optimal social choice functions: A utilitarian view. Artificial
  Intelligence  227,  190--213 (2015)

\bibitem{YuCheng}
Cheng, Y., Dughmi, S., Kempe, D.: Of the people: {V}oting is more effective
  with representative candidates. In: EC (2017)

\bibitem{Clearwater15}
Clearwater, A., Puppe, C., Slinko, A.: Generalizing the single-crossing
  property on lines and trees to intermediate preferences on median graphs. In:
  Proceedings of the 24th International Conference on Artificial Intelligence.
  pp. 32--38. IJCAI'15 (2015)

\bibitem{facilityLocationFeldman}
Feldman, M., Fiat, A., Golomb, I.: On voting and facility location. In:
  Proceedings of the 2016 ACM Conference on Economics and Computation. pp.
  269--286. EC '16, ACM, New York, NY, USA (2016)

\bibitem{DeliberativeDemocracy}
Fishkin, J., Luskin, R.: Experimenting with a democratic ideal: Deliberative
  polling and public opinion. Acta Politica  40(3),  284–298 (2005)

\bibitem{GargKGMM17}
Garg, N., Kamble, V., Goel, A., Marn, D., Munagala, K.: Collaborative
  optimization for collective decision-making in continuous spaces. In: Proc.
  World Wide Web (WWW) Conference (2017)

\bibitem{GoelLee}
Goel, A., Lee, J.: Towards large-scale deliberative decision-making: {S}mall
  groups and the importance of triads. In: ACM EC (2016)

\bibitem{Anilesh}
Goel, A., Krishnaswamy, A.K., Munagala, K.: Metric distortion of social choice
  rules. In: EC (2017)

\bibitem{GrossAX17}
Gross, S., Anshelevich, E., Xia, L.: Vote until two of you agree: Mechanisms
  with small distortion and sample complexity. In: Proceedings of the
  Thirty-First {AAAI} Conference on Artificial Intelligence, February 4-9,
  2017, San Francisco, California, {USA.} pp. 544--550 (2017)

\bibitem{Harsanyi}
Harsanyi, J.C.: A simplified bargaining model for the n-person cooperative
  game. International Economic Review  4(2),  194--220 (1963)

\bibitem{Hylland}
Hylland, A., Zenkhauser, R.: A mechanism for selecting public goods when
  preferences must be elicited. Working Paper  (1980)

\bibitem{KalaiProportionalSolutions}
Kalai, E.: Proportional solutions to bargaining situations: Interpersonal
  utility comparisons. Econometrica  45(7),  1623--1630 (1977)

\bibitem{KalaiSmorodinskyBargaining}
Kalai, E., Smorodinsky, M.: Other solutions to nash's bargaining problem.
  Econometrica  43(3),  513--518 (1975)

\bibitem{Knuth}
Knuth, D.E.: The Art of Computer Programming: Combinatorial Algorithms, Part 1.
  Addison-Wesley Professional, 1st edn. (2011)

\bibitem{Krishna}
Krishna, V., Serrano, R.: Multilateral bargaining. The Review of Economic
  Studies  63(1),  61--80 (1996)

\bibitem{Xia-chapter}
Lang, J., Xia, L.: Voting in combinatorial domains. In: Brandt, F., Conitzer,
  V., Endriss, U., Lang, J., Procaccia, A.D. (eds.) Handbook of Computational
  Social Choice. Cambridge University Press (2016)

\bibitem{convergenceIterative}
Lev, O., Rosenschein, J.S.: Convergence of iterative voting. In: Proceedings of
  the 11th International Conference on Autonomous Agents and Multiagent
  Systems-Volume 2. pp. 611--618. International Foundation for Autonomous
  Agents and Multiagent Systems (2012)

\bibitem{convergencePlurality}
Meir, R., Polukarov, M., Rosenschein, J.S., Jennings, N.R.: Convergence to
  equilibria in plurality voting. In: Proc. of 24th Conference on Artificial
  Intelligence (AAAI-10). pp. 823--828 (2010)

\bibitem{MyersonComparableUtility}
Myerson, R.B.: Two-person bargaining problems and comparable utility.
  Econometrica  45(7),  1631--1637 (1977)

\bibitem{NashBargaining}
Nash, J.F.: The bargaining problem. Econometrica  18(2),  155--162 (1950),
  \url{http://www.jstor.org/stable/1907266}

\bibitem{NeelinBargaining}
Neelin, J., Sonnenschein, H., Spiegel, M.: {An Experimental Test of
  Rubinstein's Theory of Bargaining}. Working Papers 587, Princeton University,
  Department of Economics, Industrial Relations Section (May 1986)

\bibitem{NehringP07}
Nehring, K., Puppe, C.: The structure of strategy-proof social choice. part i:
  General characterization and possibility results on median spaces. Journal of
  Economic Theory  135(1),  269 -- 305 (2007)

\bibitem{equilibriumIterative}
Rabinovich, Z., Obraztsova, S., Lev, O., Markakis, E., Rosenschein, J.S.:
  Analysis of equilibria in iterative voting schemes. In: AAAI. vol.~15, pp.
  1007--1013. Citeseer (2015)

\bibitem{RothBargaining}
Roth, A.E.: Bargaining phenomena and bargaining theory. In: Roth, A.E. (ed.)
  Laboratory Experiments in Economics: {S}ix Points of View, pp. 14--41.
  Cambridge University Press (1987)

\bibitem{RubinsteinBargaining}
Rubinstein, A.: Perfect equilibrium in a bargaining model. Econometrica  50(1),
   97--109 (1982), \url{http://www.jstor.org/stable/1912531}

\bibitem{SabanSM12}
Saban, D., Stier-Moses, N.: The Competitive Facility Location Problem in a
  Duopoly: Connections to the 1-Median Problem, pp. 539--545 (2012)

\bibitem{SchummerV}
Schummer, J., Vohra, R.V.: Strategy-proof location on a network. Journal of
  Economic Theory  104(2),  405 -- 428 (2002)

\bibitem{DeliberativeDemocracy2}
Thompson, D.F.: Deliberative democratic theory and empirical political science.
  Annual Review of Political Science  11(1),  497--520 (2008)

\bibitem{varianceProcaccia}
Wajc, D., Procaccia, A., Zhang, H.: Approximation-variance tradeoffs in
  mechanism design (January 2017)

\bibitem{WendellM81}
Wendell, R.E., McKelvey, R.D.: New perspectives in competitive location theory.
  European Journal of Operational Research  6(2),  174--182 (1981)

\end{thebibliography}
